\documentclass[a4paper,11pt]{article}
\pdfoutput=1
\usepackage{amsmath,amsthm,xcolor,amssymb}
\usepackage{bbm}
\usepackage{fullpage}
\usepackage{graphicx}
\usepackage{soul}

\usepackage{natbib}

\newcommand{\bledit}[1]{{#1}}
\newcommand{\blcomment}[1]{{}}

\newtheorem{theorem}{Theorem}[section]
\newtheorem{corollary}[theorem]{Corollary}
\newtheorem{lemma}[theorem]{Lemma}
\newtheorem{claim}[theorem]{Claim}

\newtheorem{definition}{Definition}
\newtheorem{example}{Example}

\newcommand{\rjc}[1]{{\color{orange}{#1}}}

\newcommand{\eps}{\epsilon}

\newcommand{\reals}{\mathbb{R}}

\newcommand\hide[1]{}

\DeclareMathOperator*{\argmax}{arg\,max}

\title{Non-quasi-linear Agents in Quasi-linear Mechanisms}

\author{Moshe Babaioff\thanks{Microsoft Research. 
		Email: moshe@microsoft.com} \and Richard Cole\thanks{New York University (NYU). 
		Email: cole@cs.nyu.edu. This work was supported in part by NSF Grant CCF-1909538.} \and Jason Hartline\thanks{Northwestern University. 
		Email: hartline@eecs.northwestern.edu. This work was supported in part by NSF Grant CCF-1618502.} \and Nicole Immorlica\thanks{Microsoft Research. 
		Email: nicimm@microsoft.com } \and Brendan Lucier\thanks{Microsoft Research. 
		Email: brlucier@microsoft.com  }}

\begin{document}

\maketitle

\begin{abstract}
Mechanisms with money are commonly designed under the assumption that agents are quasi-linear, meaning they have linear disutility for spending money.  
We study the implications when agents with non-linear (specifically, convex) disutility for payments participate in mechanisms designed for quasi-linear agents. We first show that any mechanism that is truthful for quasi-linear buyers has a simple best response function for buyers with non-linear disutility from payments, in which each bidder simply scales down her value for each potential outcome by a fixed factor, equal to her target return on investment (ROI).  We call such a strategy ROI-optimal.  We prove the existence of a Nash equilibrium in which agents use ROI-optimal strategies for a general class of allocation problems. Motivated by online marketplaces, we then focus on simultaneous second-price auctions for additive bidders and show that all ROI-optimal equilibria in this setting achieve constant-factor approximations to suitable welfare and revenue benchmarks.	
\end{abstract}

\clearpage
\pagenumbering{arabic}

\section{Introduction}

A fundamental assumption in much of the literature on mechanisms with money is that agents have \emph{quasi-linear} utilities.  This assumption states that agents have a linear disutility for spending money. 
Quasi-linearity is a reasonable assumption when the larger economic environment provides a \emph{constant} outside option value for money.
Yet there are many reasons why such an assumption might be violated.  For example, agents might have different liquidity constraints, or different opportunities for spending leftover cash.  In such cases, agents could exhibit heterogeneous convex disutilities for spending money.
Mechanisms designed to optimize some objective (e.g., revenue or welfare) under the assumption of quasi-linearity might fail to do so in such cases.
%
%

This paper studies properties of mechanisms designed for quasi-linear agents when the quasi-linearity assumption fails.\footnote{Mechanisms for quasi-linear agents do not generally allow agents to directly express the non-quasi-linearity of their utility.  Thus, an incentive-compatible mechanism for quasi-linear agents, i.e., where it is a dominant strategy for an agent to truthfully reveal her valuation function, will not generally remain incentive compatible for agents whose preferences are non-quasi-linear.}  We show that, if agents have convex disutility for spending money, then mechanisms that are incentive compatible for quasi-linear agents have equilibria with an attractive form.  Furthermore, in some important settings these equilibria have good revenue and welfare properties with respect to an appropriate benchmark.  These results suggest that designing mechanisms with quasi-linear agents in mind might be reasonable even if the agents actually have non-linear disutility from spending money.


The models considered in the paper are an abstraction for online marketplaces -- eBay, Booking.com, AdAuctions, Uber, etc.\ -- where shortlived users are matched to strategic long-lived agents \citep[see, e.g.,][]{hartline2019dashboard}.  A market mechanism can be viewed as selling the matching of users to agents.  Agents in these marketplaces -- sellers on eBay, hotels on Booking.com, advertisers in Ad Auctions, drivers in Uber, etc.\ -- tend to have a high volume of transactions per unit time, e.g., days or months, and continuous and stochastic models are appropriate \citep[see, e.g.,][]{athey2010structural}.  Agents often have a choice between marketplaces and this choice can result in a non-quasi-linearity. Agents can sell eclectic items on eBay or Etsy, can rent out their rooms on Booking.com or Expedia, can advertise on Google or Bing Ads, can drive for Uber or Lyft, etc.  The ad auctions of Yandex and Facebook are especially relevant to the questions of this paper as both employ mechanisms that are incentive compatible for agents with quasi-linear preferences.


Informally, our first main result shows that deploying a mechanism that \emph{would be} incentive compatible for quasi-linear agents yields one or more Nash equilibria with an attractive structure. 
Each (possibly non-quasi-linear) agent always has a utility-maximizing report in which they adjust their true valuation by a scalar multiplier.  The multiplier can be interpreted as a single-dimensional tuning parameter that an agent can use to control spend.
At a utility-maximizing report, the multiplier will be equal to the reciprocal of the marginal return on investment (ROI) the agent obtains from the mechanism in their utility-maximizing declaration.
\bledit{We call such a strategy \emph{ROI-optimal}}.  We show that there is at least one pure Nash equilibrium in which all agents use ROI-optimal strategies.\footnote{We emphasize that this is a Nash equilibrium of the full information game, where the bids of each player are a best response to the bids of the others for the given valuations.  Equilibrium existence can depend on tie-breaking; we elaborate on this below.}

We next analyze the efficiency and revenue properties of these equilibria.  We note that one cannot hope to approximate the utilitarian welfare (i.e., the sum of buyer and seller utilities) of the welfare-maximizing allocation in this setting, even in very special cases such as selling a single item.  
This is because 
from the perspective of any mechanism, the behavior of any agent is invariant under any scaling of their utilities, paired with an equal scaling of their disutility for spending money.\footnote{For example, while the behavior of  a quasi-linear buyer with value 1 for an item will be the same as the behavior of a buyer with value $100$ and cost of money $100p$ when spending $p$, maximizing the utilitarian welfare requires the item be allocated to the second buyer. 
}  
It is therefore not possible to directly compare the utilities of different agents and perform the tradeoffs required to maximize utilitarian welfare.
It is in fact not even meaningful to sum utilities that are incomparable and hence utilitarian welfare is nonsensical in our setting.  
See Section~\ref{sec:example.util.welfare} for an example illustrating this difficulty.


We propose an alternative metric for welfare that we argue is more relevant to a platform that interacts with agents through payments.  We define the \emph{transferable welfare} of an agent for an allocation $x$ as the maximum amount that the agent would be willing to pay for $x$, given a take-it-or-leave-it offer.  
Note that this metric is invariant under the operation of scaling an agent's values and disutility for spending money.  It generalizes the notion of ``liquid welfare'' from the literature on budget-constrained buyers~\citep{Dobzinski2014liquid}.  It also corresponds to the definition of welfare commonly used for a monopolist interacting with a quasi-linear buyer, which is measured in dollars precisely because of an assumed equivalence between an agent's utility for an outcome and their willingness to pay. 
Another thing to note is that the transferable welfare generated by a mechanism is the total amount of money, in dollars, that the population of buyers
would be willing to pay in order to implement the allocation of the mechanism.\footnote{Alternatively, 
this is the total amount of money, in dollars, that the population of agents (buyers as well as the seller)
would be willing to pay in order to implement the outcome of the mechanism (allocation as well as payments), given that the seller utility is quasi-linear.}  In this sense, the optimal transferable welfare
corresponds to the maximum possible revenue that could be collected by any mechanism, ignoring incentive constraints beyond individual rationality. 

Solution concept in hand, we turn our attention to the canonical problem of selling multiple divisible goods to additive buyers.  As a reminder, even when the agents' values for goods are additive, agent preferences may be still be non-quasi-linear with respect to payments. In this setting, simultaneous second-price auctions would be truthful for quasi-linear bidders, so our equilibrium characterization 
applies to that mechanism.
We show that, at any {ROI-optimal} equilibrium of simultaneous second-price auctions, the transferable welfare generated is at least half of the optimum. 
Thus, while the total utilitarian welfare can be very different across equilibria, the transferable welfare is approximately optimal at each {ROI-optimal} equilibrium.
{This approximation factor is tight: we provide an example in which the transferable welfare at any {ROI-optimal} equilibrium is indeed half of the optimal transferable welfare.}

We also present a bound on the revenue achieved at {ROI-optimal equilibria}.  The revenue generated by simultaneous second-price auctions at equilibrium can be quite low, which (in practice) encourages the use of reserve prices.  We show that this mechanism, paired with appropriate reserve prices, leads to revenue at {every} {ROI-optimal} equilibrium that approximates the revenue of a corresponding prior-free posted-price benchmark.  Specifically, our benchmark is the maximum possible revenue that could be generated by static sequential posted pricing, where each item is assigned a fixed price and then agents are approached in a revenue-maximizing order and asked to purchase their utility-maximizing sets.  We prove that if the corresponding prices are used as reserves in the second-price auctions, then the revenue at any {ROI-optimal} equilibrium is at least half of this revenue benchmark, and this analysis is tight. Put another way, for any choice of agent valuations and outside option values, the equilibrium behavior generated by non-uniform value for money does not unduly
reduce revenue relative to what could be achieved by interacting with agents individually and sequentially via posted pricing.  In this sense, the revenue generated by the simultaneous second-price auctions with reserves is robust to equilibrium selection {when focusing on ROI-optimal equilibria}.  

One important subtlety in our results is the handling of ties.  The allocation rule of the mechanism can, in general, be discontinuous.  This is the case for the VCG mechanism, for example.  If so, the existence of pure Nash equilibrium can depend delicately on tie-breaking.  Issues with ties can arise even with quasi-linear agents\footnote{Consider a single-item first price auction with two bidders, where bidder $1$ has value $1$ and bidder $2$ has value $2$.  Then a pure Nash equilibrium exists only if ties are broken in favor of bidder $2$.\label{foot::NEexist}} but the presence of non-linearities adds an extra level of delicacy.
In practice this is arguably not a concern: we show that if the interim allocation and payment rules are continuous then the existence of pure equilibria is guaranteed regardless of how ties are broken.  Such continuity might arise due to noise in the mechanism itself, such as the natural variability across impressions in an online advertising platform.  But if discontinuities cannot be ruled out, we show that our pure Nash equilibria always exist if one appropriately endogenizes the tie-breaking rule. 
The proof follows by interpreting the discontinuous allocation rule as a limit of carefully-crafted continuous rules, then using the resulting convergent sequence of equilibria to define a tie-breaking rule for the original mechanism.

\subsection{Related Work}

Our work is related to the literature on mechanism design with
non-quasi-linear agents.  One commonly-studied form of
non-quasi-linearity is risk aversion in buyers.  In this literature,
buyers typically have single-dimensional values and have a concave
utility for wealth, and the goal is to characterize revenue-optimal
mechanisms~\citep{maskin1984optimal} or compare the revenue and
equilibrium properties of common auction formats~\citep{HuMZ2010}.
\citet{fu2013prior} show that first-price auctions are approximately
revenue-optimal for single-dimensional buyers with a form of concave
utilities corresponding to capped quasi-linearity.
\citet{GreenwaldOS2018} consider the
impact of convex cost of money curves in the context of
single-dimensional buyers, and demonstrate that all-pay auctions are
approximately revenue-optimal for certain polynomial cost functions;
their utility model is closest to our own.  Our buyers are also risk
averse (i.e., have concave utility for wealth), but have
multi-dimensional types.  \bledit{\citet{alaei2013simple} show
that revenue maximization in a multi-dimensional multi-agent setting 
reduces (exactly or approximately) to a single-agent 
problem in a variety of settings, including risk-averse agents.}

One special case of concave utility is agents with hard budget
constraints, and this setting has been studied in the literature for
multi-dimensional buyers.  In general it is impossible to guarantee
Pareto optimality in auctions with unknown (private) budget
constraints, though some auction formats (such as clinching auctions)
can achieve Pareto optimality when budgets are publicly
known~\citep{Dobzinski2012budget}.  On the other hand, it is possible
to achieve a constant approximation to the optimal revenue in the
presence of private budgets~\citep{abrams2006revenue,borgs2005multi}
with an approximation that degrades with the budget-dominance
parameter, and revenue-optimal mechanisms can be characterized for a
single multi-dimensional agent with a
public~\citep{laffont1996optimal} or private budget
\citep{pai2014optimal}.  
\citet{feng2020simple} combine the approach of
\citet{alaei2013simple} with the analysis of \citet{abrams2006revenue}
to approximately reduce mechanism design for non-linear agents to
single-dimensional mechanism design for linear agents.  The mechanisms
they contruct are cognizant of the non-linearity of the agents, whereas 
the mechanisms we study are not.

\bledit{While Pareto optimality cannot be guaranteed in the presence of private budget constraints}, there is a line of
work studying the ``liquid welfare'' of a mechanism, defined to be the
total value of an allocation in which each buyer's valuation is capped
at their budget.  As it turns out, the optimal liquid welfare can be
approximated for budget-constrained buyers in various
settings~\citep{Dobzinski2014liquid,Azar2017liquid}.  Our notion of
transferable welfare coincides with liquid welfare for
budget-constrained agents, and can be viewed as an extension to more
general utility models.

Our work is also closely related to the literature on budget smoothing in advertising auction platforms, which focuses on throttling or pacing the spend of budget-constrained agents so that their budget is expended uniformly over a given time period.  This framework is mechanically very similar to ours in the particular case that buyers are quasi-linear up to a hard budget.  \citet{Balseiro2015mean} studied mean-field equilibria of repeated second-price auctions with budget-constrained single-dimensional agents, and found the throttling is approximately utility-maximizing.  Later work extended to many different budget management strategies (again with single-dimensional agents) and compared their resulting equilibria both empirically and analytically~\citep{Balseiro2017budget}.  Algorithmically, \citet{mehta2007adwords} initiated a long line of literature studying the optimal allocation for budgeted buyers in online settings.

Perhaps closest to our work is that of \citet{Conitzer2018SPA}, who study budget smoothing in the case of additive preferences over multiple divisible goods and introduce the notion of a pacing equilibrium, where the platform shades bids on behalf of the budget-constrained agents.  They demonstrate that when the items are sold by second-price auctions, equilibria exist and may not be unique, and such equilibria can have poor welfare or revenue properties.  Moreover, it is NP-hard to find a revenue- or welfare-maximizing equilibrium, though they provide numerical evidence that equilibria are not difficult to find in practice.  \citet{Conitzer2019FPA} then study a variant with first-price auctions, under the assumption that values must be scaled uniformly (which is no longer necessarily utility-optimal).  They show that the resulting outcomes correspond to a form of market equilibrium and can be computed efficiently.  Our results apply to a more general model of agent utility, and we show that equilibria in the second-price auction variant achieve good approximations to suitable welfare and revenue benchmarks (which differ from the measures studied by \citet{Conitzer2018SPA}). 

There is also a line of literature that directly considers advertiser ROI in online ad auctions.  Golrezaei, Lobel and Paes Leme provide evidence that advertisers face minimum-ROI constraints, and design optimal auctions for single-dimensional agents under such constraints~\citep{Golrezaei2018ROI}.  Our buyers are multi-dimensional and we model incentives with respect to marginal ROI rather than average ROI targets.  \citet{borgs2007dynamics} consider dynamic auctions in which agents optimize by equating ROI across goods, and show that adding perturbations to the auctions can ease convergence to an equilibrium outcome under first-price payment rules.

\section{Model}


\paragraph{Problem Setting}

There is one seller with linear value for money.  There are $n$ agents and a space $X$ of potential outcomes.
Agent $i$ has a valuation function $v_i : X \to {\reals_{\geq 0}}$.  We'll write $V$ for the space of allowable valuation functions, and we will assume that $V$ is closed under scalar multiplication.  We write $v = (v_1, \dotsc, v_n)$ to denote the profile of valuations.  We will assume that the outcome space $X$ is convex; one way to satisfy this condition is to allow randomization and define the value of a lottery to be the expected value of the resulting outcome. 
{For notational convenience we write $x_i$ for the part of outcome $x$ that is payoff-relevant to agent $i$, so that \bledit{$v_i(x_i) = v_i(x)$} 
for all $i$.  We refer to $x_i$ as the allocation to agent $i$.}
This general formulation covers many common mechanism design problems, including the following:
\begin{itemize}
\item \textbf{Combinatorial auctions of indivisible goods.}  There is a set $G$ of indivisible goods, and an outcome is an  allocation that assigns a disjoint subset of goods $x_i \subseteq G$ to each agent $i$.  A valuation function $v_i$ assigns a value to each potential subset $x_i$ that $i$ might receive.
\item \textbf{Additive valuations over divisible goods.}  There are $m$ divisible goods and each agent's valuation is linear in the quantity of good received and additive across goods.  In this case an allocation is given by $x = (x_{ij})$ where $x_{ij} \in [0,1]$ is the fraction of good $j$ allocated to agent $i$, and the valuation of agent $i$ for allocation $x$ is $v_i(x) = \sum_j v_{ij} x_{ij}$ where $v_{ij} \geq 0$ is agent $i$'s value per unit of good $j$.
\item \textbf{Public projects.}  There is a set $X$ of projects.  Each agent has a valuation $v_i(x) \geq 0$ for each project $x \in X$.  A single project is to be selected and each agent will receive their value for the chosen project.  {Note that in this setting we have $x_i = x$ for all $i$, since the entire outcome is payoff-relevant for each agent.}
\end{itemize}
Going forward, we will be using the terms agent, buyer, and bidder interchangeably.

\paragraph{Value for Money}
Each buyer has an outside option value for money, described by a bound $B_i\in \reals_{> 0} \cup \{\infty\}$ 
and a function $C_i \colon [0,B_i] \to R_{\geq 0}$. The buyer cares about their spending {\em in expectation}.\footnote{We can alternatively imagine that the buyer participates in many auctions, and their incentives and spending constraints bind in aggregate over many instances.}
A finite value of $B_i$ represents a hard budget constraint at $B_i$, meaning the expected payment may not exceed $B_i$;  and $B_i=\infty$ means there is no budget constraint. If $p_i$ is the expected payment of buyer $i$, then
$C_i(p_i)$ is the utility that buyer $i$ foregoes by paying $p_i\leq B_i$ to the seller. 
$C_i$ is a strictly increasing weakly convex function, continuous on $[0,B_i]$ and normalized so that $C_i(0) = 0$.  In particular this implies that $C_i$ is invertible on $[0,B_i]$. 
For example, in the quasilinear setting, $B_i=\infty$ and $C_i(p_i) = p_i$ for all $p_i\geq 0$, while a standard hard budget constraint of $B$ is represented by $B_i=B$ and $C_i(p_i) = p_i$ for $p_i \leq B$.
When $C_i$ is differentiable, we will write $R_i(p_i)$ for the derivative of $C_i$ at $p_i$.
Without loss of generality, by separately rescaling each buyer's value units, we can assume that $R_i(0) = 1$ for every $i$.\footnote{That is, the first {(infinitesimal)} unit of money is worth one {(infinitesimal)} unit of value.}  This implies that for any $p_i$, the minimum outside option value for money for bidder $i$ at $p_i$ is at least $1$. %
It is also convenient to scale the unit of money (and hence each bidder's value unit)
so that for every buyer, the bundle of all the goods has value at most 1.\footnote{It may appear here that we've lost generality by scaling twice.  But the first scaling choice is simply
setting the ratios of each bidder's value unit to the money unit; the second choice is equivalent to specifying the unit of measurement for money, and hence for each bidder, the unit of measurement for its value.} 

We define buyer $i$'s \emph{willingness to pay} for an allocation $x_i$ 
to be $W_i(x_i) = C_i^{-1}(v_i(x_i))$, and the \emph{transferable welfare} of outcome $x$ to be $W(x) = \sum_i W_i(x_i)$.  Note that this quantity is well-defined, as $C_i$ is strictly increasing on $[0, B_i]$ and $v_i(x_i) < \infty$ for all $x_i$.  \bledit{For notational convenience we define $C_i(p) = \infty$ for $p > B_i$ and $C_i^{-1}(v) = B_i$ for $v > C_i(B_i)$, so that $W_i(x_i) = B_i$ whenever $v_i(x_i) > B_i$.}

\blcomment{I split the example into 3 parts.  The typo in the example has been fixed (utilities in the symmetric equilibrium are 1.5, not 1).}

\begin{example}\label{example:basic}
To illustrate the model, suppose there are two divisible items to allocate among two buyers.  Valuations are additive, where $v_{ij}$ denotes buyer $i$'s value per unit of good $j$.
These values are given by $v_{11} = 2, v_{12} = 1, v_{21} = 1$, and $v_{22} = 2$.\footnote{{To keep the example simple, we have not normalized the valuations and the money unit.}}  Each buyer has a hard budget, with budgets given by $B_1 = B_2 = 0.5$.  The buyers are quasi-linear up to their budgets: $C_i(p) = p$ for $p \leq B_i$.  
One possible allocation is $x=(x_1,x_2)$ with $x_1=(1,0)$ and $x_2=(0,1)$; that is, buyer 1 gets item 1 and buyer 2 gets item 2. The value of this allocation for each buyer $i$ is $v_i(x_i)=2$.  Note that this is the allocation that maximizes the sum of agent values.  Each agent's willingness to pay for their allocation is $W_i(x_i) = 0.5$, due to the hard budget constraints.
\end{example}

\paragraph{Utility and Bids} A direct-revelation mechanism is described by an allocation rule $x \colon V^n \to X$ and payment rule $p \colon V^n \to \reals^n$ that map a profile of valuation reports (or \emph{bids}) to an outcome and a profile of payments.  Given a fixed mechanism $M$, we write $x_i(b)$ and $p_i(b)$ to denote the allocation to agent $i$
and the payment that agent $i$ is charged when buyers bid according to $b$.
Then, if buyer $i$'s valuation is $v_i$, its utility 
is ${U}_i(b; v_i) = {v}_i(x_i(b)) - C_i(p_i(b))$.
A pure Nash equilibrium is a profile of bids $b$ such that each buyer is simultaneously best-responding given their valuation.  That is, bids $b = (b_1, \dotsc, b_n)$ form a pure Nash equilibrium for valuation profile $v$ if, for each $i$ and each possible bid $b'_i$,
$U_i(b_i, b_{-i}; v_i) \geq U_i(b'_i, b_{-i}; v_i)$.
We are interested in proving the existence of equilibria (and characterizing their properties) when bidders are using a particular form of bidding strategy, which is to choose some $\alpha_i \in [0,1]$
and then set $b_{i} = \alpha_i v_{i}$.  That is, each buyer scales their valuation function by a uniform scaling factor $\alpha_i \in [0,1]$. 
We call such a bid a \emph{uniform scaling}.  
We write $\alpha = (\alpha_i)$ for a profile of scaling factors.  
{We denote by $x_i(\alpha \cdot v) {= x_i(\alpha_1 v_1, \alpha_2 v_2,\ldots, \alpha_n v_n)}$ the allocation of agent $i$ when agents each bid according to a scaling strategy given by 
vector $\alpha$ of scaling factors}.

\begin{example}\label{example:basic2}
Consider again the setting of Example~\ref{example:basic}, and suppose the goods are being sold using the VCG mechanism.  Since valuations are additive, this is equivalent to simultaneous second-price auctions: each agent places a bid on each good, and the highest-bidding agent wins the good and pays the second-highest bid. 
Suppose the buyers bid their true valuations: $b_1 = (2,1)$ and $b_2 = (1,2)$.  Then each buyer $i$ will receive all of good $i$ (for $i \in \{1,2\}$) and pay $1$.  As this exceeds the buyers' budgets, truthful reporting is not an equilibrium.
If each buyer $i$ instead bids \emph{half} of their true valuation, applying a uniform scaling strategy with $\alpha_i = 0.5$, then the allocation remains unchanged
but each buyer will instead pay $0.5$, exhausting their budget. The utility of every buyer is $1.5$, and one can check that neither buyer can increase their utility by bidding differently.
	This is therefore a pure Nash equilibrium in uniform scaling strategies, which is symmetric (in scaling factors) and results in a value-maximizing allocation. 
\end{example}

Pure Nash equilibria need not be unique, even when restricting to equilibria where agents use uniform scaling strategies.  Moreover, the set of equilibria can depend on the way ties are broken, as the next example shows.

\begin{example}\label{example:basic3}
Consider once again the setting of Example~\ref{example:basic}.
There can be other equilibria of the VCG mechanism beyond the one described in Example~\ref{example:basic2}, where buyers scale their valuations by different factors.
Consider the case that $\alpha= (\alpha_1,\alpha_2)= (1/3,2/3)$,
	and each buyer $i$ is bidding $\alpha_i v_i$, so buyer $1$ bids $b_1=(2/3,1/3)$ and buyer $2$ bids $b_2=(2/3,4/3)$. 
	The price of item $1$ is $2/3$, and item $2$ is priced at $1/3$. Note that there is a tie for item $1$, so item $1$ can be split between the buyers (and each buyer will pay $2/3$ per unit received), while item $2$ will be fully allocated to buyer $2$.
	Consider the allocation $y_1=(3/4,0)$ and $y_2=(1/4,1)$, which divides item $1$ in a particular way.  Then each buyer pays $0.5$ and exhausts their budget. One can check that no buyer can gain by bidding differently.
	Of course, one can reverse the roles of the buyers and choose $(\alpha_1,\alpha_2) = (2/3, 1/3)$.  This will also be a pure Nash equilibrium under appropriate tie-breaking.  As it turns out, these are the only outcomes that can arise as a pure Nash equilibrium.  So the set of equilibria is discrete, finite, and discontinuous, and different equilibria result in different value profiles for the buyers.
\end{example}

\subsection{A Motivating Example: Platform Competition}
\label{sec:model.competition}


As a motivating example for a non-linear value for money spent, consider an advertising platform in a competitive environment in which other platforms are present.  
Suppose each buyer has a certain total budget to spend on advertising across all platforms. This buyer derives value from different allocations (i.e., impressions) on our platform, represented by a valuation function $v$.  Fixing the behavior of others, each possible report of the buyer results in a payment $p$ and allocation $x$ (and hence value $v(x)$) from our platform's mechanism.  Given that the agent spends $p$ on the platform, she will seek to maximize her expected value over all strategies with expected payment $p$ (including possibly randomizing between reports).  The mechanism therefore describes a frontier of potentially optimal (payment, {value}) outcomes, which will be a concave curve $\overline{V}$ that maps payment to {value}; see Figure~\ref{fig:outside.option}(a).

\begin{figure}[t]
\centering
\includegraphics[width=1.0\textwidth]{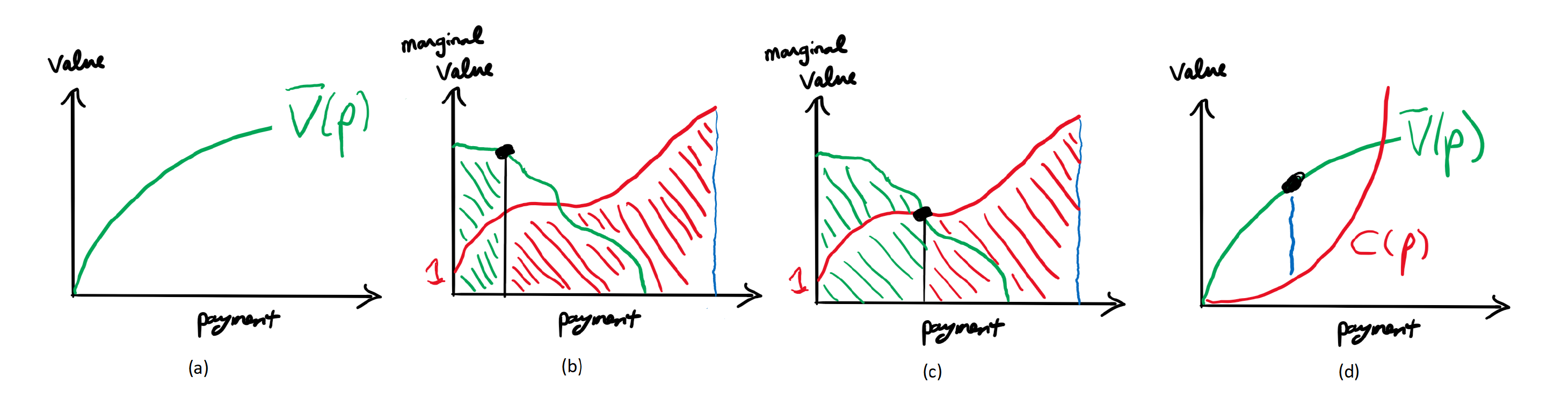}
\caption{\small Visualizing our model for a single agent.  (a) By plotting (payment, {value}) for different reports to a mechanism, one gets a concave frontier $\overline{V}(p)$ of undominated mechanism outcomes.  (b) Taking the derivative of $\overline{V}$ gives the marginal value per unit of spend (green curve), which can be compared with the marginal value lost from outside options (red curve).  A given level of spend (black dot) generates value from the mechanism's allocation (green area) and from outside options (red area).  (c) The optimal level of spend equates the marginal values (black dot), maximizing total value. (d) By integrating, we can interpret the opportunity cost as a convex cost $C(p)$ of spending $p$, and the optimal bid maximizes the difference $\overline{V}(p) - C(p)$.}
\label{fig:outside.option}
\end{figure}

The concavity of $\overline{V}$ has a natural interpretation: the marginal value gained by the agent from the mechanism is weakly decreasing in her spend, as she will always choose the ``best bang per buck'' for each additional dollar spent.  But the agent can also derive value from her outside option(s), where the marginal value per dollar is likewise weakly decreasing. 
The agent therefore maximizes her total value (from all sources) by equating the marginal value obtained from the mechanism with the marginal value obtained from the outside option(s),\footnote{The marginal value obtained per unit of spend, $\Delta v / \Delta p$, is simply $1$ plus the marginal ROI.  Equating marginal values is therefore equivalent to equating marginal ROI.} as this is the point where the opportunity cost of spending more money on the mechanism outweighs any additional value that can be obtained from the platform.  See Figure~\ref{fig:outside.option}(c).  This opportunity cost defines a convex value for money: each additional dollar spent on the platform pulls from increasingly higher-value uses of money on alternative channels.  Each buyer's outside options can therefore be summarized as a convex function $C(p)$ representing her disutility for spending $p$ on the mechanism; see Figure~\ref{fig:outside.option}(d).


\section{Best Responses and Structure of Equilibria}
\label{sec:equil.exact}




We study equilibria of the game that is induced when the seller deploys a mechanism $M$ which would be incentive compatible for quasi-linear agents.
We will show that for any profile of bids, each agent has a best response which is a uniform scaling strategy.  Moreover,
{at an equilibrium of such scaling strategies,}
 each agent's scaling factor corresponds to {her} 
 \emph{marginal value for money} at the equilibrium level of spending.

\subsection{Best Responses and Uniform Scaling}
\label{sec:best.response}



Our first result is that for any agent $i$ and any bids $b_{-i}$ of the other agents, there is always a utility-maximizing bid of the form $\alpha_i v_i$ for some $\alpha_i \in [0,1]$.  Moreover, this $\alpha_i$ will equal the inverse of the agent's marginal cost of money, evaluated at their total payment $p_i(\alpha_i v_i, b_{-i})$.  We call such an $\alpha_i$ ROI-optimal.

\begin{definition}
\label{defn::ROIoptimal}
For a given buyer $i$ with valuation $v_i$ and bids $b_{-i}$ of the other buyers, a uniform scaling bid $b_i=\alpha_i\cdot v_i$ is \emph{ROI-optimal} for agent $i$ if $1/\alpha_i$ is a subderivative of $C_i$ at $P_i$, where $P_i = p_i(b_i, b_{-i})$ is agent $i$'s payment at bid $b_i$.
\end{definition}

Note that if $C_i$ is everywhere differentiable, then the condition for ROI-optimality can be restated as $\alpha_i = 1 / R_i(p_i(b_i, b_{-i}))$.

\begin{theorem}
	\label{thm:best.response.general}
	Fix any mechanism that is truthful for quasi-linear buyers.  Fix any buyer $i$ with valuation $v_i$ and any declared valuations $b_{-i}$ of the other buyers.
	Then there is a uniform scaling bid $b_i = \alpha_i \cdot v_i$ with $\alpha_i \in [0,1]$ 
	such that, under appropriate tie-breaking, 
	$b_i \in \argmax_{b_i}\{ U_i(b) \}$
	and $\alpha_i$ is ROI-optimal.  Moreover, any ROI-optimal choice of $\alpha_i$ is utility-maximizing.  
\end{theorem}

The full proof of Theorem~\ref{thm:best.response.general} appears in Appendix~\ref{sec:equil.exact.ties.details}.
As a warm-up, we will present a simplified proof that relies on three additional assumptions we'll make for technical convenience. 

First, we assume that agents do not have hard budgets.  That is, $B_i = \infty$ for all $i$, and hence $C_i(p)$ is well-defined and finite for all $p \geq 0$.

Second, we assume that the mechanism in the statement of Theorem~\ref{thm:best.response.general} admits continuous allocation and payment rules.  More precisely, for any agent $i$ with valuation $v_i$, and any declared valuations $b_{-i}$ of the other agents, we will assume that $v_i(x_i( \alpha_i \cdot v_i, b_{-i}))$ and $p_i( \alpha_i \cdot v_i, b_{-i} )$ are both continuous functions of $\alpha_i \in [0,1]$.  The purpose of this assumption is to avoid subtleties around tie-breaking, since a potential tie between bids (e.g., in an auction mechanism) would correspond to a discontinuity in the allocation and payment rules.  Thus, under this assumption, {ties occur with zero probability} 
and can be safely ignored. 

Third, we'll assume that for any agent $i$ and any declared valuations $b_{-i}$ of the other agents, agent $i$'s payment $p_i( \alpha_i \cdot v_i, b_{-i} )$ is strictly increasing in $\alpha_i \in [0,1]$.  The role of this assumption is to rule out indifferences where an agent faces the same outcome and payment at different choices of uniform scaling factor $\alpha_i$.  This will allow us to additionally show that the utility-maximizing choice of $\alpha_i$ is unique (and hence every utility-maximizing $\alpha_i$ is ROI-optimal).  


%
We now proceed with the proof {assuming these assumptions hold}.  Choose buyer $i$ and fix the bids $b_{-i}$ of the other buyers.  Consider the collection of all possible bids $b_i$ of buyer $i$, 
each of which generates a payment and a value $(p_i(b_i, b_{-i}),v_i(x_i(b_i, b_{-i})))$.  Take the collection of all such points, as plotted with payment on the $x$-axis and value on the $y$-axis, and define $\overline{V}_i(p)$ to be the upper concave envelope of these points.  That is, $\overline{V}_i(p)$ is the maximum possible value that buyer $i$ can obtain by paying $p$, possibly in expectation over some randomization over bids.\footnote{We will show in Claim~\ref{claim:value.per.spend} that each point on the concave envelope actually corresponds to a choice of bid, so no randomization is necessary.}  Recall Figure~\ref{fig:outside.option}.  If we write $\overline{p}$ for the maximum possible payment (which must be finite as the values are bounded), then $\overline{V}_i$ is defined on the range $[0,\overline{p}]$.  Note that $\overline{V}_i$ is continuous, 
concave, weakly increasing, and $\overline{V}_i(0) = 0$.  
%
%
The maximum possible utility attainable by buyer $i$ when paying $p$ is $\overline{V}_i(p) - C_i(p)$.  We claim that for each $p \in [0,\overline{p}]$, the value of $\overline{V}_i(p)$ is achievable with a uniform scaling bid.

\begin{claim}
\label{claim:value.per.spend}
Choose any $p \in [0,\overline{p}]$, any subderivative $r$ of $\overline{V}_i$ at $p$, and take $\alpha_i = 1/r$.  Then $p = p_i(\alpha_i v_i, b_{-i})$ and $V_i(\alpha_i v_i, b_{-i}) = \overline{V}_i(p)$.
\end{claim}
\begin{proof}
%
Consider an alternative outside option value curve $\widetilde{C}_i$ given by $\widetilde{C}_i(p) = r \cdot p$ for all $p$.  That is, $\widetilde{C}_i$ is a constant value of $r$ per unit of money spent.  Strategically, a buyer with valuation $v_i$ and outside option curve $\widetilde{C}_i$ behaves equivalently to a quasi-linear buyer with valuation $v_i / r$.  That is, we can re-scale the buyer's valuation to express it in units of money rather than units of utility.  But now, since the underlying mechanism is assumed to be truthful for quasi-linear buyers,
we know that such a buyer must maximize utility by declaring valuation $v_i / r$.  In other words, the uniform scaling bid $v_i / r$ must select a utility-maximizing point on curve $\overline{V}_i(\cdot)$, say $(p, \overline{V}_i(p))$.  As this point is utility-maximizing, $\overline{V}_i(p)$ must have a subderivative $r$ (as this is the derivative of $\widetilde{C}_i$ everywhere).  

We next claim that there is a unique $p \in [0,\overline{p}]$ at which $r$ is a subderivative of $\overline{V}_i$.  Suppose not.  This means (by concavity) that the set of all such points form a contiguous non-empty interval, say $[\beta,\gamma]$.
Changing from declaring $\alpha_i = 1/r-\delta$ to $\alpha_i = 1/r+\delta$, for any $\delta > 0$, then leads to a positive jump in the value $\overline{V}_i$, as it results in changing a payment of at most $\beta$ to a payment of at least $\gamma > \beta$ and $\overline{V}_i$ is continuous and increasing.  But this contradicts the assumption that the allocation and payment rules of the mechanism are continuous. 

We conclude that $p$ is the unique point in $[0,\overline{p}]$ at which $r$ is a subderivative of $\overline{V}_i$ at $p$.  
As the uniform scaling bid $\alpha_i = 1/r$ selects this point, we must have $p = p_i(\alpha_i v_i, b_{-i})$ and $V_i(\alpha_i v_i, b_{-i}) = \overline{V}_i(p)$ as claimed.
\end{proof}

\begin{corollary}
\label{cor:strict.incr.payment}
$\overline{V}_i$ is strictly concave on $[0,\overline{p}]$.
Also, if the payment $p_i(\alpha_i v_i, b_{-i})$ is strictly increasing in $\alpha_i$ then $\overline{V}_i$ is 
differentiable on $[0,\overline{p}]$.
\end{corollary}
\begin{proof}
By Claim~\ref{claim:value.per.spend}, $\overline{V}_i$ can have subderivative $r$ only at $p = p_i(v_i / r, b_{-i})$, so $\overline{V}_i$ is strictly concave.
If $p_i(\alpha_i v_i, b_{-i})$ is strictly increasing then 
$\alpha_i < \alpha'_i$ implies $p_i(\alpha_i v_i, b_{-i}) < p_i(\alpha'_i v_i, b_{-i})$.
Claim~\ref{claim:value.per.spend} therefore implies that $\overline{V}_i$ has a unique subderivative, and is therefore differentiable, at each $p \in [0,\overline{p}]$.
\end{proof}


We conclude that there is a unique $p \in [0,\overline{p}]$ that maximizes $\overline{V}_i(p) - C_i(p)$, and for $r = R_i(p)$ there exists a bid $b_i = \alpha_i \cdot v_i$ such that $r = 1/\alpha_i$ and $p = p_i(b)$.  Any such $b_i$ is ROI-optimal for agent $i$ by definition, 
and is utility-maximizing from the choice of $p$.  Moreover, if the payment of agent $i$ is strictly increasing in $\alpha_i$ then $r$ is the unique subderivative of $\overline{V}_i$ at $p$, and therefore the only utility-maximizing choice of $\alpha_i$ is $1/r$, which is ROI-optimal.
Finally, since we assume that the derivative of $C_i$ is at least $1$ everywhere, we must have $r \geq 1$ at any utility-maximizing payment $p$, which implies that $\alpha_i = 1/r \leq 1$.  This completes the proof of Theorem~\ref{thm:best.response.general}, under our simplifying  technical assumptions.

\subsection{Equilibrium Existence}\label{sec:eq-exist}

A direct corollary of 
Theorem~\ref{thm:best.response.general}
is that if a bid profile $b$ is simultaneously ROI-optimal for all agents 
then it is utility-maximizing for all agents and therefore forms a pure Nash equilibrium.  
We'll say that a pure Nash equilibrium is ROI-optimal if all agents are using ROI-optimal strategies.  Our next result is that at least one such ROI-optimal pure Nash equilibrium is guaranteed to exist.  That is, there exists a pure Nash equilibrium in which all agents use ROI-optimal uniform scaling strategies.  Note that although we are focusing attention on uniform scaling strategies, our definition of equilibrium is with respect to the broader space of bidding strategies in which each agent can declare any value for any item.  

\begin{theorem}
\label{thm:equil.proper.general}
Fix any mechanism that is truthful for quasi-linear buyers.
Then for any given market instance there exists an ROI-optimal pure Nash equilibrium in uniform scaling strategies under appropriate tie-breaking.
%
\end{theorem}

The full proof of Theorem~\ref{thm:equil.proper.general} appears in Appendix~\ref{sec:equil.exact.ties.details}.  As we did for Theorem~\ref{thm:best.response.general}, we will now present a simplified proof of Theorem~\ref{thm:equil.proper.general} that holds under the same technical assumptions: each $C_i$ is everywhere continuously differentiable, the mechanism has continuous allocation and payment rules, and the payment rule is strictly increasing.
The proof follows from standard fixed-point arguments, by showing that utilities are quasi-concave in the scaling parameter.

\begin{proof}
We use a theorem of Debreu, Glicksberg, and Fan for existence of equilibria in infinite games: there exists a pure Nash equilibrium for any quasi-concave game with compact and convex strategy space and continuous payoffs.  We verify each of these conditions for our game, restricted to uniform scaling strategies.  The strategy space of each buyer is $[0,1]$, which is compact and convex. 
Recall that we assume our mechanism has continuous allocation and payment rules, and moreover 
by assumption
the cost function $C_i$ is continuous in payment for each agent.  
Payoffs are therefore continuous.  

We now show that $U_i(\alpha_i, \alpha_{-i})$ is quasi-concave in $\alpha_i$. Fix $\alpha_{-i}$, and 
%
recall from Claim~\ref{claim:value.per.spend} that $V_i(\alpha_i, \alpha_{-i}) = \overline{V}_i(p_i(\alpha_i, \alpha_{-i}))$.
Then 
\begin{equation}
\label{eq.alpha.p}
U_i(\alpha_i, \alpha_{-i}) = \overline{V}_i(p_i(\alpha_i, \alpha_{-i})) - C_i(p_i(\alpha_i, \alpha_{-i})).
\end{equation}
Recall that $C_i(p)$ is strictly increasing and convex,
and the curve $\overline{V}_i$ is concave by construction.  As $p_i(\alpha_i)$ is increasing in $\alpha_i$, $U_i$ is quasi-concave as required, so there exists a pure equilibrium of the game in which agents declare scaling factors.

It remains to show that any pure Nash equilibrium in uniform scaling strategies must be ROI-optimal, and that $\alpha$ is an equilibrium in the more general bidding game.  If each agent faces a strictly increasing pricing rule, each agent's utility-maximizing choice of $\alpha_i$ is unique and ROI-optimal 
as argued in Section~\ref{sec:best.response}.
As the choice of $\alpha_i$ is unique, it must correspond to the value of $\alpha_i$ from Theorem~\ref{thm:best.response.general}, so in particular it is also utility-maximizing over all possible bids of agent $i$ (not just uniform scaling strategies).  Thus $\alpha$ is a pure Nash equilibrium even with respect to non-uniform-scaling strategies.
\end{proof}

\paragraph{A note about tie-breaking.}
Both Theorem~\ref{thm:best.response.general} and Theorem~\ref{thm:equil.proper.general} hold under appropriate tie-breaking.  The simplified proofs provided in this section circumvent tie-breaking issues by imposing extra technical assumptions that ensure that ties occur with zero probability.  The full proofs require {careful} handling of ties.  We make this formal in Appendix~\ref{sec:equil.exact.ties.details}.  Roughly speaking, whenever the allocation rule of the mechanism is discontinuous, we refer to the outcome chosen at the point of discontinuity as being selected by a tie-breaking rule.  Different tie-breaking rules can select different outcomes from the convex closure of the limit points at the discontinuity.  Theorem~\ref{thm:equil.proper.general} shows that an ROI-optimal Nash equilibrium exists for \emph{some} choice of tie-breaking rule, and this tie-breaking rule might depend on the market instance.\footnote{Such ``endogenized tie-breaking'' is a common requirement for the existence of pure Nash equilibrium in mechanisms even for quasi-linear bidders.  See footnote~\ref{foot::NEexist} for an illustrating example.
}

\section{Properties of Equilibria: Additive and Divisible Goods}
\label{sec:equil.properties}

Now that we have established the existence of ROI-optimal equilibria,\footnote{As noted in Section~\ref{sec:equil.exact}, the existence result depends on the nature of tie-breaking.  The results in this section apply to any ROI-optimal equilibrium that arises under any choice of tie-breaking rule.  See Appendix~\ref{sec:equil.exact.ties.details} for a more thorough treatment of tie-breaking.} we can prove approximate welfare and revenue guarantees at any such equilibrium in the central setting of selling multiple divisible goods to buyers with additive values, using simultaneous second-price auctions.
In this model, an allocation to buyer $i$ is given by $x_i = (x_{ij})$ where $x_{ij} \in [0,1]$ is a fraction of good $j$.  The valuation of buyer $i$ for allocation $x_i$ is $v_i(x_i) = \sum_j v_{ij} x_{ij}$, where $v_{ij} \geq 0$ is buyer $i$'s value per unit of good $j$.  Recall that under additive valuations, the simultaneous second-price auction is equivalent to the VCG mechanism and is truthful for quasi-linear buyers.

%

\subsection{Transferable Welfare at Equilibrium}

We first establish a $2$-approximation for transferable welfare at any ROI-optimal pure Nash equilibrium of the simultaneous second-price auction with no reserves.  
\bledit{Recall from Examples~\ref{example:basic2} and~\ref{example:basic3} that multiple ROI-optimal equilibria may exist; our bounds will apply in the worst case over all such equilibria.}
Recall that buyer $i$'s \emph{willingness to pay} for an  allocation ${x}_i$ is ${W_i}({x}_i) = C_i^{-1}(v_i({x}_i))$, and the \emph{transferable welfare} of allocation ${x}$ is $W({x}) = \sum_i {W_i}({x}_i)$.
We will compare the transferable welfare achieved under allocation $x$ to the
transferable welfare {$W(y)=\sum_i W_i(y_i)$}
achieved under an arbitrary allocation $y$.


%

\begin{theorem}
\label{thm:W-approx}
Suppose $\alpha$ is an ROI-optimal pure Nash equilibrium of the second-price auction (without reserves) with allocation ${x}$.
Then the transferable welfare of allocation $x$ is at least half the maximum transferable welfare: $W({x}) \geq \frac{1}{2} \max_{y} {W(y)} $. 
\end{theorem}

\begin{proof}
Write $P_i = p_i(\alpha)$ for the total spend of buyer $i$ at equilibrium $\alpha$.
We'll write $p_j$ for the second-highest bid on item $j$.  That is, the second-highest element of the set $\{v_{ij}\alpha_{i} \colon 1 \leq i \leq n\}$.

Let $y$ be an arbitrary allocation.  Choose some agent $i$, and consider two cases based on whether buyer $i$ is willing to pay more for allocation $y_i$ than allocation $x_i$.  

%

\noindent
{\bf Case 1:} $W_i(y_i) > W_i(x_i)$.

Note that if $v_i(y_i) < v_i(x_i)$ then $W_i(x_i) > W_i(y_i)$, so it must be that $v_i(y_i) \geq v_i(x_i)$.
Recalling that $W_i(x_i) = C^{-1}(v_i(x_i))$ and $W_i(y_i) = C^{-1}(v_i(y_i))$, the convexity of $C_i$ implies that
\[ \frac{v_i(y_i) - v_i(x_i)}{W_i(y_i) - W_i(x_i)} 
=\frac{C(W_i(y_i)) - C(W_i(x_i))}{W_i(y_i) - W_i(x_i)} 
\geq R^+_i(W_i(x_i)) \]
recalling that $R^+_i$ denotes the right-derivative of $C_i$.
This implies
\begin{align*}
W_i(y_i) - W_i(x_i) 
& \leq \frac{v_i(y_i) - v_i(x_i)}{R^+_i(W_i(x_i))}.
\end{align*}

Since $x_i$ is the allocation obtained at equilibrium, $P_i \leq W_i(x_i)$ (for otherwise buyer $i$ would obtain negative utility).  This implies $R^+_i(W_i(x_i)) \geq R^+_i(P_i)$, since $C_i$ is weakly convex.  Combining this with the inequality above we conclude that
\[ W_i(y_i) - W_i(x_i) \leq \frac{v_i(y_i) - v_i(x_i)}{R^+_i(P_i)}.\]

If $y_{ij} > x_{ij}$ then $x_{ij} < 1$, which means that agent $i$'s bid on item $j$ is at most the second-highest.  This means that the price on item $j$ (which is the second-highest bid in the second-price auction without reserves) is $p_j \geq v_{ij} \alpha_i$.  The ROI-optimality of $\alpha_i$ implies that $\alpha_i \geq 1/R_i^+(P_i)$, and thus $p_j \geq v_{ij} / R^+_i(P_i)$.
This justifies the final inequality in the equation below:
\begin{align*}
W_i(y_i) - W_i(x_i) \leq \frac{v_i(y_i) - v_i(x_i)}{R^+_i(P_i)}
&\leq \frac{\sum_j v_{ij}(y_{ij} - x_{ij})^+}{R^+_i(P_i)}\\
&\leq \sum_{j} p_j (y_{ij} - x_{ij})^+.
\end{align*}

\noindent
{\bf Case 2:} $W_i(y_i) \leq W_i(x_i)$. 
%
In this case $W_i(y_i) - W_i(x_i)  \leq 0 \leq \sum_{j} p_j (y_{ij} - x_{ij})^+$.

\medskip

Summing up over all $i$ in both cases, we conclude that
\begin{align*}
W(y) - W(x) \leq \sum_{i, j} p_j \cdot (y_{ij} - x_{ij})^+ \leq \sum_j p_j.
\end{align*}

Since the market clears at equilibrium (as all reserves are $0$ by assumption) every item $j$ is completely sold and is generating total revenue of $p_j$, so $\sum_j p_j$ is the total revenue generated from all buyers, $\sum_i P_i$.  Additionally, as argued above, we have $P_i \leq W_i(x_i)$ for all $i$. Therefore 
\begin{align*}
W({y}) - W({x}) \leq \sum_j p_j = \sum_i P_i \leq \sum_i W_i(x_i) = W({x})
\end{align*}
and hence $W({x}) \geq \frac{1}{2} W({y})$ as required.
\end{proof}

In Appendix~\ref{sec:lowerbound.welfare} we show that this result is tight and there exists a setting in which 
the transferable welfare of the allocation of an ROI-optimal pure Nash equilibrium is indeed only half the maximum transferable welfare, and not more than that.

\subsection{Equilibrium Revenue}
Next we establish a revenue bound for ROI-optimal equilibria.  Our benchmark will be related to the revenue achievable by sequential posted pricing.  
Given a profile of item prices and an ordering of the buyers, a \emph{sequential posted price} mechanism approaches the buyers one by one in order and allows each to purchase a utility-maximizing outcome from the remaining supply.
We show that if the seller can achieve a certain revenue using sequential posted prices (and an arbitrary order of the buyers, presumably chosen to maximize the resulting revenue), then at least half of that revenue is
achieved at any competitive market equilibrium with reserves set equal to those posted prices.  


\begin{theorem}
\label{thm:Rev-approx}
Suppose $\alpha$ is an ROI-optimal pure Nash equilibrium of the second-price auction with item reserves $(r_j)$, resulting in allocation $x$ and total revenue $Rev(\alpha)$.
Let ${y}$ be the allocation resulting from sequential posted prices with price vector $(r_j)$ and an arbitrary arrival order.  
Then $Rev(\alpha) \geq \frac{1}{2} \sum_{ij} y_{ij} r_j$.
\end{theorem}

\begin{proof}
Fix any arrival order for the sequential posted prices mechanism with price vector $(r_j)$, and denote the resulting revenue by 
$Rev(y,r)=\sum_{ij} y_{ij} r_j$.
Write $P_{i}$ for the total spend of buyer $i$ in equilibrium $\alpha$. With this notation,  $Rev(\alpha)= \sum_i P_i$; our goal is to prove that $Rev(\alpha) = \sum_i P_i \geq \frac{1}{2} Rev(y,r)$.

Write $p_j$ for the second-highest 
bid on item $j$ (under bid profile $\alpha$) or $r_j$ (the reserve price for item $j$), whichever is higher.  That is, $p_j$ is the price paid for item $j$. 
Define $x_j = \sum_i x_{ij}$ to be the total quantity of item $j$ sold under allocation $x$, and similarly for $y_j$.  
Fix some $j$, and note that if $x_j \geq y_j$, then since $p_j \geq r_j$, necessarily $\sum_{i} x_{ij}p_j \geq \sum_i y_{ij}r_j$.  Similarly, if $x_j < y_j$ but $x_{ij} \geq y_{ij}$, then  $x_{ij}p_j \geq y_{ij}r_j$.

From this we obtain the following bound on the revenue at equilibrium, covering only the cases described above.

\begin{align*}
Rev(\alpha) 
&= \sum_{i,j} x_{ij}p_j\\
& \geq \sum_{j \colon x_j \geq y_j} x_j p_j + \sum_{j \colon x_j < y_j} \sum_{i \colon x_{ij} \geq y_{ij}} x_{ij} p_j \\
& \geq \sum_{j \colon x_j \geq y_j} y_j r_j + \sum_{j \colon x_j < y_j} \sum_{i \colon x_{ij} \geq y_{ij}} y_{ij} r_j
\end{align*}

Now consider some item $j$ with $x_{j} < y_j$, and some buyer $i$ for which $x_{ij} < y_{ij}$ (of which there must be at least one).  
Write $P_i(y,r)$ for the total payment of buyer $i$ under allocation $y$ and prices $(r_j)$. Since $x_j < y_j \leq 1$, not all of item $j$ is sold in allocation $x$, which means that $p_j = r_j$.  

Since $x_{ij} < 1$, agent $i$'s bid on item $j$ is either at most the second-highest {or it equals the reserve price}.

On the other hand, under sequential posted pricing with prices $(r_j)$, buyer $i$ chose to purchase $y_{ij}$ units of item $j$ (from a possibly restricted subset of items).  This provided a value of $v_{ij}$ per unit of item $j$, at a cost of $r_j$ per unit.  We must therefore have that 
$R^-_i(P_i(y,r)) \leq v_{ij} / r_j$, since the marginal cost of purchasing the item must be at most the marginal gain.

Combining inequalities, $v_{ij} / R^-_i(P_i(y,r)) \geq r_j \geq v_{ij} / R^+_i(P_i)$, and so $R^-_i(P_i(y,r))\leq R^+_i(P_i)$.  We conclude that $P_i(y, r) \leq P_i$, where we used the fact that $R^-(q_1) \leq R^+(q_2)$ implies $q_1 \leq q_2$. 

We are now ready to derive our second bound on the revenue at equilibrium.

\begin{align*}
Rev(\alpha)
= \sum_i P_i 
& \geq \sum_{i} \mathbf{1}[\exists j \colon x_j < y_j, x_{ij} < y_{ij}] P_i \\
& \geq \sum_{i} \mathbf{1}[\exists j \colon x_j < y_j, x_{ij} < y_{ij}] P_i(y,r) \\
& \geq \sum_{j \colon x_j < y_j} \sum_{i \colon x_{ij} < y_{ij}} y_{ij} r_j.
\end{align*}

Combining these two bounds yields
\begin{align*}
2 Rev(\alpha) \geq \sum_{i,j} y_{ij} r_j = Rev(y,r),
\end{align*}
as required.
\end{proof}


{Note that Theorem~\ref{thm:Rev-approx} holds pointwise for any given realization of preferences, and for any reserve prices. 
It therefore also holds in expectation in a Bayesian setting.}

In Appendix \ref{sec:lowerbound.revenue} we show that our analysis is tight. We show that there exists a setting and reserve prices in which 
	the total revenue $Rev(\alpha)$ of some ROI-optimal pure Nash equilibrium $\alpha$ with these reserves
is indeed only half the revenue achieved with sequential posted prices for some order of the players when pricing at these reserves, and not more than that.

%

\section{Application: Sponsored Search Advertising}
\label{sec:equil.spa}


To illustrate our results, we will consider a particular instance of our model motivated by online sponsored search advertising.  The seller is selling $m$ divisible goods, with the quantity of each good normalized to $1$ unit.  We think of each good as representing a day's worth of impressions on a given keyword, which we treat as a large quantity.
Each buyer $i$ has value for each impression of good $j$, that value is at most $v_{ij} \geq 0$ and varies stochastically. A random variable $\gamma_{ij} \in [0,1]$ represents an impression-specific variation in value for buyer $i$ and good $j$, yielding an impression-specific value 
	of $v_{ij} \cdot \gamma_{ij}$.
This variation captures, for example, {searcher}-specific click rate, conversion rate, intent, etc., as estimated by the platform at the time of impression {based on {searcher} attributes, which vary over the searcher population. Note that in the special case that $\gamma_{ij} = 1$ for sure for all $i$ and $j$, this reduces to the case of additive valuations over divisible goods.

We use $v_i= (v_{i1} \dotsc, v_{im})$ to denote the vector of item values for agent $i$, $v = (v_1, \dotsc, v_n)$ to denote the full profile of values, and $\gamma = (\gamma_{ij})$ to denote the matrix of modifiers. Here $v_{i}$ is buyer $i$'s private valuation, 
and $\gamma$ is a random variable in $[0,1]^{nm}$ drawn from a known distribution.  We evaluate buyer value and payments in expectation over $\gamma$.
Buyer $i$'s value per unit for good $j$ is then $v_{ij} \cdot \gamma_{ij}$, and his value is additive across goods. Note that this is equivalent to drawing a new instance of $\gamma$ for each impression and then evaluating value and payments in aggregate over a day's worth of impressions (in the limit as the number of impressions grows large).  This is natural in the context of sponsored search: we expect advertisers to evaluate the disutility of their payments over many auctions, and not for each individual keyword search.

Consider the VCG mechanism, which in this setting has a particularly simple form.
Each buyer $i$ places a bid $b_{ij}\geq 0$ on each item $j$, and items are then sold by individual per-item second price auctions.
Bids are placed before the modifiers $\gamma$ are realized.  Once bids are placed, $\gamma$ is realized and bid $b_{ij}$ is multiplied by $\gamma_{ij}$ before being entered into the auction.  That is, the effective bid of agent $i$ on item $j$ is $b_{ij} \gamma_{ij}$.\footnote{In the context of ad auctions, this corresponds to interpreting $b_{ij}$ as a bid per impression, which the platform combines with an estimate of match quality (e.g., searcher-specific variation in attention, click rate, etc.) at the time of impression.}  We remind the reader that agent $i$'s {maximal value $v_{ij}$} 
for item $j$ is likewise multiplied by $\gamma_{ij}$, so a bid {vector} 
$b_{i} = v_{i}$ reflects agent $i$'s true valuation regardless of the realization of $\gamma$.

If the VCG mechanism is used to resolve individual keyword auctions, then by Theorem~\ref{thm:equil.proper.general} there will be an ROI-optimal pure Nash equilibrium in scaling strategies.  This means that each advertiser {can} best-respond by bidding a scaled version of their true value profile.  In general this equilibrium depends on appropriate tie-breaking, but if we furthermore assume that each $\gamma_{ij}$ is drawn independently from an atomless distribution then ties will occur with probability $0$.  The allocation and payment rules of the mechanism will therefore be continuous (in expectation over $\gamma$).

By Theorem~\ref{thm:W-approx}, the transferable welfare of any such equilibrium will be at least half of the optimal transferable welfare, regardless of each advertiser's budget and/or disutility $C_i$ for spending money.  Moreover, by Theorem~\ref{thm:Rev-approx}, for any setting of reserve prices, the revenue obtained at {such} an equilibrium will be at least half of the revenue obtained in an environment where advertisers simply take their most desired quantities of keywords at the reserve prices.

\bibliographystyle{apalike}
\bibliography{a}

\begin{thebibliography}{}

\bibitem[Abrams, 2006]{abrams2006revenue}
Abrams, Z. (2006).
\newblock Revenue maximization when bidders have budgets.
\newblock In {\em Proceedings of the seventeenth annual ACM-SIAM symposium on
  Discrete algorithm}, pages 1074--1082.

\bibitem[Alaei et~al., 2013]{alaei2013simple}
Alaei, S., Fu, H., Haghpanah, N., and Hartline, J. (2013).
\newblock The simple economics of approximately optimal auctions.
\newblock In {\em 2013 IEEE 54th Annual Symposium on Foundations of Computer
  Science}, pages 628--637.

\bibitem[Athey and Nekipelov, 2010]{athey2010structural}
Athey, S. and Nekipelov, D. (2010).
\newblock A structural model of sponsored search advertising auctions.
\newblock In {\em Sixth ad auctions workshop}, volume~15.

\bibitem[Azar et~al., 2017]{Azar2017liquid}
Azar, Y., Feldman, M., Gravin, N., and Roytman, A. (2017).
\newblock Liquid price of anarchy.
\newblock In {\em Algorithmic Game Theory}, pages 3--15.

\bibitem[Balseiro et~al., 2017]{Balseiro2017budget}
Balseiro, S., Kim, A., Mahdian, M., and Mirrokni, V. (2017).
\newblock Budget management strategies in repeated auctions.
\newblock In {\em Proceedings of the 26th International Conference on World
  Wide Web}, pages 15--23.

\bibitem[Balseiro et~al., 2015]{Balseiro2015mean}
Balseiro, S.~R., Besbes, O., and Weintraub, G.~Y. (2015).
\newblock Repeated auctions with budgets in ad exchanges: Approximations and
  design.
\newblock {\em Management Science}, 61(4):864--884.

\bibitem[Borgs et~al., 2007]{borgs2007dynamics}
Borgs, C., Chayes, J., Immorlica, N., Jain, K., Etesami, O., and Mahdian, M.
  (2007).
\newblock Dynamics of bid optimization in online advertisement auctions.
\newblock In {\em Proceedings of the 16th international conference on World
  Wide Web}, pages 531--540.

\bibitem[Borgs et~al., 2005]{borgs2005multi}
Borgs, C., Chayes, J., Immorlica, N., Mahdian, M., and Saberi, A. (2005).
\newblock Multi-unit auctions with budget-constrained bidders.
\newblock In {\em Proceedings of the 6th ACM conference on Electronic
  commerce}, pages 44--51.

\bibitem[Conitzer et~al., 2019]{Conitzer2019FPA}
Conitzer, V., Kroer, C., Panigrahi, D., Schrijvers, O., Sodomka, E.,
  Stier-Moses, N.~E., and Wilkens, C. (2019).
\newblock Pacing equilibrium in first-price auction markets.
\newblock In {\em Proceedings of the 2019 {ACM} Conference on Economics and
  Computation}, page 587. {ACM}.

\bibitem[Conitzer et~al., 2018]{Conitzer2018SPA}
Conitzer, V., Kroer, C., Sodomka, E., and Stier-Moses, N.~E. (2018).
\newblock Multiplicative pacing equilibria in auction markets.
\newblock In {\em Web and Internet Economics - 14th International Conference,
  {WINE} 2018, Proceedings}, volume 11316, page 443.

\bibitem[Dobzinski et~al., 2012]{Dobzinski2012budget}
Dobzinski, S., Lavi, R., and Nisan, N. (2012).
\newblock Multi-unit auctions with budget limits.
\newblock {\em Games and Economic Behavior}, 74(2):486 -- 503.

\bibitem[Dobzinski and Paes~Leme, 2014]{Dobzinski2014liquid}
Dobzinski, S. and Paes~Leme, R. (2014).
\newblock Efficiency guarantees in auctions with budgets.
\newblock In {\em ICALP}.

\bibitem[Feng et~al., 2020]{feng2020simple}
Feng, Y., Hartline, J., and Li, Y. (2020).
\newblock Simple mechanisms for non-linear agents.
\newblock {\em arXiv preprint arXiv:2003.00545}.

\bibitem[Fu et~al., 2013]{fu2013prior}
Fu, H., Hartline, J., and Hoy, D. (2013).
\newblock Prior-independent auctions for risk-averse agents.
\newblock In {\em Proceedings of the fourteenth ACM conference on Electronic
  commerce}, pages 471--488.

\bibitem[Golrezaei et~al., 2018]{Golrezaei2018ROI}
Golrezaei, N., Lobel, I., and Paes~Leme, R. (2018).
\newblock Auction design for {ROI}-constrained buyers.
\newblock Available at SSRN: https://ssrn.com/abstract=3124929.

\bibitem[Greenwald et~al., 2018]{GreenwaldOS2018}
Greenwald, A., Oyakawa, T., and Syrgkanis, V. (2018).
\newblock On revenue-maximizing mechanisms assuming convex costs.
\newblock In {\em Algorithmic Game Theory - 11th International Symposium,
  {SAGT} 2018, Proceedings}, pages 113--124.

\bibitem[Hartline et~al., 2019]{hartline2019dashboard}
Hartline, J.~D., Johnsen, A., Nekipelov, D., and Zoeter, O. (2019).
\newblock Dashboard mechanisms for online marketplaces.
\newblock In {\em Proceedings of the 2019 ACM Conference on Economics and
  Computation}, pages 591--592.

\bibitem[Hu et~al., 2010]{HuMZ2010}
Hu, A., Matthews, S.~A., and Zou, L. (2010).
\newblock Risk aversion and optimal reserve prices in first- and second-price
  auctions.
\newblock {\em Journal of Economic Theory}, 145(3):1188 -- 1202.

\bibitem[Laffont and Robert, 1996]{laffont1996optimal}
Laffont, J.-J. and Robert, J. (1996).
\newblock Optimal auction with financially constrained buyers.
\newblock {\em Economics Letters}, 52(2):181--186.

\bibitem[Maskin and Riley, 1984]{maskin1984optimal}
Maskin, E. and Riley, J. (1984).
\newblock Optimal auctions with risk averse buyers.
\newblock {\em Econometrica: Journal of the Econometric Society}, pages
  1473--1518.

\bibitem[Mehta et~al., 2007]{mehta2007adwords}
Mehta, A., Saberi, A., Vazirani, U., and Vazirani, V. (2007).
\newblock Adwords and generalized online matching.
\newblock {\em Journal of the ACM (JACM)}, 54(5):22--es.

\bibitem[Pai and Vohra, 2014]{pai2014optimal}
Pai, M.~M. and Vohra, R. (2014).
\newblock Optimal auctions with financially constrained buyers.
\newblock {\em Journal of Economic Theory}, 150:383--425.

\end{thebibliography}

\appendix
\section{Examples}
\label{sec:examples}


\subsection{Utilitarian Welfare Cannot be Approximated}
\label{sec:example.util.welfare}

We'll say the utilitarian welfare of a mechanism is the total net utility generated by the mechanism.  That is, the difference between (a) the sum of utilities obtained by the mechanism (including the seller who is assumed to have quasi-linear utility), and (b) the sum of utilities obtained by the null mechanism that makes no allocation and takes no payments.

We begin by showing that, for the second-price auction, there are cases where the total net utility generated by the mechanism does not approximate the optimal net utility (over all outcomes).
Suppose there are two buyers and a single item.  Buyer $1$ has value $v_1 = 2$ and outside option for money given by $C_1(p) = p$.  Buyer $2$ has value $v_2 = 150$ and $C_2(p) = 100p$.  Given that a second-price auction is used, it is a dominant strategy for buyer $1$ to bid $2$ (their true value) and for buyer $2$ to bid $150/100 = 1.5$.  Thus buyer $1$ will win the item for a price of $1.5$.

Note that the utilitarian welfare of this outcome is $2$, as buyer $1$ obtains a value of $2$, and the utility lost by spending $1.5$ is gained by the seller.
%
On the other hand, the outcome that allocates the item to buyer $2$ and collects no payment generates a net utility of $150$.  By scaling up the value and outside option of buyer $2$ this gap can be made arbitrarily large.

This example is driven by the observation that an agent with value $v_i$ and cost curve $C_i(p) = p$ is strategically indistinguishable from one with value $v_i \cdot r$ and $C_i(p) = p \cdot r$, for any $r > 0$.  That is, those agents have the same set of preferences over outcomes and payments that might be returned by the mechanism.  Using this observation we next extend this example to one in which no mechanism can obtain a sublinear approximation 
to the optimal utilitarian welfare.

There are $n$ buyers.  Each of the buyers has value $v_i = 1$ and outside option $C_i(p) = p$, except for one buyer $i^*$ (indexed uniformly at random) who has value $v_{i^*} = R$ and $C_i(p) = R \cdot p$ for some large $R > 1$.  As the mechanism cannot distinguish buyer $i^*$ from the other buyers, the mechanism allocates $x_{i^*} \leq 1/n$.  Thus the net utility generated is at most $R/n + 1$.  On the other hand, the outcome that allocates $x_{i^*} = 1$ and charges no payments generates a net utility of $R$, for a gap of $\Omega(n)$ as $R$ grows large.

%

\subsection{Non-ROI-Optimal equilibria with Smoothed Allocation Rules}
\label{sec:example.non.ROI.optimal}

We now present an example with a uniform scaling equilibrium that is not ROI-optimal, and for which the transferable welfare does not approximate the maximum transferable welfare.  This example will fall within the model from Section~\ref{sec:equil.spa}, and holds even if values are perturbed by random variables $\gamma = (\gamma_{ij})$ drawn from a distribution that is atomless but does not have full support.  

Suppose there are two buyers and one item.  Buyer $1$ has $v_1 = 1$, and is quasi-linear up to a budget of $\delta > 0$.  We think of $\delta$ as being very small.  Buyer $2$ has value $v_2 = 1/2$ and is quasi-linear.  Suppose that each $\gamma_{ij}$ is drawn uniformly from $[1/2, 1]$.

This example admits the following non-ROI-optimal equilibrium: buyer $1$ bids $1$ and buyer $2$ bids $0$. 
Note that even after multiplying each of the bids by $\gamma_{ij}$ buyer $1$ would want to bid at least $1/2$, while buyer 2 would want to bid 0.
Thus buyer $1$ obtains the entirety of the good for free, whereas buyer $2$ is indifferent between all possible bids in $[0, 1/2]$ as all such bids result in the zero allocation  (and bidding above $1/2$ is dominated).  However, the transferable welfare of this outcome is only $\delta$, as this is the maximum amount that buyer $1$ is willing to pay for the good.  A transferable welfare of $1/2$ is achieved by allocating the good to buyer $2$.  As $\delta > 0$ is arbitrary, this gap can be made as large as desired.  

Note that this example demonstrates a  failure to approximate the transferable welfare for $\gamma$ that is atomless, yet this result does not  {hinge on the atomless property.} 
If we instead removed the perturbations by setting $\gamma_{ij} = 1$ deterministically for each $i$ and $j$, then the equilibrium described above remains an equilibrium and the gap in transferable welfare remains arbitrarily large.

\subsection{Multiplicity of Equilibria}\label{sec:multiple-NE}

We next show that there can be multiple ROI-optimal equilibria, and different equilibria can result in different transferable welfare.

Suppose there are two buyers and two items.  Values are given by $v_{11} = 2$, $v_{12} = 1$, $v_{21} = 1$, and $v_{22} = 2$.  That is, buyer $i$ values item $i$ at $2$ and the other item at $1$.  Each buyer has $C_i(p) = p$ for $p \leq 1/2$, and $C_i(p) = 1/2 + 10(p-1/2)$ for $p > 1/2$.  That is, their marginal value for money is $1$ for the first $1/2$ spent, then their marginal value for money is $10$ thereafter.  We think of $1/2$ as a soft budget limit for each player.

There is a symmetric equilibrium where each buyer applies a scaling factor of $1/2$.  In this equilibrium buyer $1$ wins all of item $1$, buyer $2$ wins all of item $2$, and each buyer exhausts their soft budget.  Each buyer obtains a value of $2$, for which they would be willing to spend up to $0.65$.  So the total transferable utility is $1.3$.

There are also two asymmetric equilibria.  Suppose buyer $1$ applies scaling factor $1/3$ for a bid profile of $(2/3, 1/3)$, and buyer $2$ applies a scaling factor of $2/3$ for a bid profile of $(2/3, 4/3)$.  Then buyer $2$ wins all of item $2$, and there is a tie on item $1$.  Assuming the tie is broken so that $x_{11} = 3/4$ and $x_{21} = 1/4$, each buyer will exhaust all of their soft budget. This is indeed an equilibrium, as neither player is willing to spend more at a marginal cost of $10$ for money.  Buyer $1$ obtains a value of $2.25$, for which they would be willing to spend $0.675$.  Buyer $2$ obtains a value of $1.5$, for which they would be willing to spend $0.6$.  So the total transferable utility is $1.275$, which is less than $1.3$.  The other asymmetric equilibrium is similar, obtained by swapping the roles of buyer $1$ and buyer $2$.

\subsection{Tightness of Approximation for Transferable Welfare and Revenue}
\label{sec:lowerbound.welfare}
\label{sec:lowerbound.revenue}

We now give an example showing that the factor $1/2$
in our transferable welfare bounds (Theorem~\ref{thm:W-approx}) and in our revenue bounds (Theorem~\ref{thm:Rev-approx}), are both tight. 
A single example will illustrate tightness for both results.

Fix some small $\epsilon > 0$.  There are two items and two buyers.  Values are given by $v_{11} = 1, v_{12} = 0, v_{21} = 1+\epsilon$, and $v_{22} = 1$.  Each buyer has a hard budget, with budgets given by $B_1 = 1$ and $B_2 = 1$.  The buyers are quasi-linear up to their budgets: $C_i(p) = p$ for $p \leq B_i$.  

First consider transferable welfare.  The optimal transferable welfare is $2$, achieved by allocating item $1$ to buyer $1$ and item $2$ to buyer $2$.  For that allocation, each agent is willing to spend their entire budget of $1$.

Suppose the agents bid with uniform scaling factors $\alpha_1 = \alpha_2 = 1$.  Then buyer $2$ wins both items outright, paying $1$ for item $1$ and $0$ for item $2$.  This is a best-response for buyer $2$, as their spend is within their budget and they are obtaining positive utility per unit of item $1$. It is also a best-response for buyer $1$, as they would not be willing to spend $1+\epsilon$ per unit of item $1$.  Thus $(\alpha_1, \alpha_2) = (1,1)$ is an equilibrium.  But the transferable welfare of the resulting allocation is $1$, as this is the most that buyer $2$ is willing to spend.  The allocation at this equilibrium therefore obtains half of the optimal transferable welfare.

Next we will consider revenue.  For each item $j$, set a reserve price of $r_j = 1$.  Suppose these are set as posted prices and offered to the buyers sequentially, first to buyer $1$ then to buyer $2$.  Then buyer $1$ would purchase item $1$, and then buyer $2$ would purchase item $2$ (as it is the only item left).  The total revenue obtained is $2$.

Next suppose the values $r_j$ are used as reserve prices in our second-price auction.  Consider the bidding profile with uniform scaling factors $\alpha_1 = \alpha_2 = 1$.  The allocation where buyer $1$ obtains no items and buyer $2$ obtains all of item $1$ is consistent with this bidding profile.  Moreover, this is an equilibrium since neither bidder wishes to reduce their bid: buyer $2$ spends their budget and is obtaining positive utility at the price for item $1$ (which is still the reserve price of $1$), and buyer $1$ is obtaining no allocation so cannot benefit by reducing their bid.  But at this equilibrium, only item $1$ is sold for a revenue of $1$. This equilibrium therefore obtains half of the revenue attainable by sequential posted pricing with prices $(1,1)$.

\section{Equilibria of the Second-Price Auction with Ties}
\label{sec:equil.exact.ties.details}


We wish to revisit Theorem~\ref{thm:best.response.general} and Theorem~\ref{thm:equil.proper.general} and prove them in full, without making any additional technical assumptions.  Our interest is in equilibria with pure strategies, which means that we must be careful about handling  ties.  For the utility of an agent can be discontinuous in their bid, so a pure equilibrium is not guaranteed to exist for arbitrary tie-breaking rules.  
Yet, we will show that there always exists a pure equilibrium of the auction game (in ROI-optimal strategies) with appropriately endogenized tie-breaking.  

Fix a mechanism $M$ with allocation rule $x(\cdot)$ and payment rule $p(\cdot)$.  The \emph{closure of $(x,p)$} is a correspondence $\overline{x}$, where $\overline{x}(v)$ is the set of all $(x^*, p^*)$ such that there exists some countable sequence $(v^i)$ with $\lim_{i \to \infty}v^i = v$, $\lim_{i \to \infty}x(v^i) = x^*$, and $\lim_{i \to \infty}p(v^i) = p^*$.  Note that $\overline{x}(v) = \{(x(v), p(v))\}$ whenever $x$ is continuous at $v$, but $\overline{x}(v)$ may contain multiple pairs if $x$ is discontinuous at $v$.  We will denote by $\overline{X}(v)$ the set of all convex combinations of elements of $\overline{x}(v)$, so that $\overline{X}$ is a convex correspondence.\footnote{Recall that the set of outcomes is assumed to be convex, for example by allowing randomization over a base set of outcomes.}

A pair of allocation and payment rules $x'$ and $p'$ such that $(x'(v), p'(v)) \in \overline{X}(v)$ for all $v$ is said to be equivalent to $(x, p)$ up to tie-breaking.  We will sometimes refer to $(x', p')$ as a tie-breaking rule for $M$ (or an allocation rule consistent with $M$ up to tie-breaking).  We'll say that two truthful mechanisms $M$ and $M'$ are \emph{equivalent up to tie-breaking} if their allocation and payment rules are equivalent up to tie-breaking.

\begin{example}
Consider a second-price auction for a single indivisible good with two buyers.  Then an allocation is a pair $(x_1, x_2)$ and a valuation profile is a pair $(v_1, v_2)$.  Suppose ties are broken in favor of buyer $1$.  Then the allocation rule is given by $x(v_1, v_2) = (1,0)$ if $v_1 \geq v_2$ and $x(v_1, v_2) = (0,1)$ if $v_1 < v_2$. Let's consider the convex closure of this allocation rule.  (We'll ignore the payment rule here to simplify notation, but payments will always follow the second-price rule.)  For any $v_1 \neq v_2$, $x$ is continuous around $(v_1, v_2)$ so $\overline{X}(v_1, v_2)$ is the singleton $\{ x(v_1, v_2) \}$.   For any $z \geq 0$, $x(z,z) = (1,0)$.  We also have $\lim_{\eps \to 0} x(z, z+\eps) = (0,1)$.  So $\overline{x}(z,z) = \{(0,1), (1,0)\}$.  Taking convex combinations, we have $\overline{X}(z,z) = \{ (p, 1-p) \colon p \in [0,1] \}$.  This naturally corresponds to all ways to break ties when both buyers bid $z$: all ways to randomize the allocation between the two buyers.  If we chose $x'$ such that $x'(v_1, v_2) = x(v_1, v_2)$ whenever $v_1 \neq v_2$ and $x'(v_1, v_2) = (0.5, 0.5)$ whenever $v_1 = v_2$, then $x'(v) \in \overline{X}(v)$ for all $v$.  This $x'$ encodes the alternative tie-breaking rule that assigns the good uniformly at random in case of a tie.
\end{example}

We will show that for any truthful mechanism $M$, there is another mechanism equivalent up to tie-breaking for which a pure Nash equilibrium exists.  Moreover, at this equilibrium all agents use ROI-optimal uniform scaling strategies, and each agent's scalar multiplier corresponds to their \emph{marginal value for money} at the equilibrium level of spending.  We discuss these refinements below.

\subsection{Best Responses and Uniform Scaling}

We will prove that for every buyer $i$ and bids of the others, there is a utility-maximizing bid $b_i$ that is ROI-optimal under some tie-breaking rule.  
%
Note that this claim by itself does \emph{not} imply the existence of a pure Nash equilibrium for some tie-breaking rule, as we may use different tie-breaking rules to establish this claim for different buyers. To prove the existence of a pure Nash equilibrium we actually need to show that all buyers are best-responding under the \emph{same} tie-breaking rule. 
Indeed, we later show (Section \ref{sec:eq-exist.ties}) that such a consistent tie-breaking rule does exist.  The following is a restatement of Theorem~\ref{thm:best.response.general} that formalizes the notion of appropriate tie-breaking.

\begin{theorem}
	\label{thm:best.response.alpha.ties}
	Fix any mechanism $M$ that is truthful for quasi-linear buyers.
	Fix any buyer $i$ with valuation $v_i$ and any declared valuations $b_{-i}$ of the other buyers.
	Then there is a uniform scaling bid $b_i = \alpha_i \cdot v_i$ with $\alpha_i \in [0,1]$ and a tie-breaking rule $(x,p)$ for $M$ such that $b_i$ is utility-maximizing for agent $i$:
	\[ b_i \in \argmax_{b_i}\{ v_i(x_i(b)) - C_i(p_i(b)) \}.\]  
Moreover, at least one such $b_i$ is ROI-optimal for agent $i$, and any ROI-optimal bid $b_i$ is utility-maximizing for agent $i$.
\end{theorem}

\begin{proof}
Choose buyer $i$ and fix the bids $b_{-i}$ of the other buyers.  Write $\overline{X}$ for the correspondence of tie-breaking rules for $M$.  Consider the collection of all possible bids $b_i$ of buyer $i$ and allocation and payment rules $(x(\cdot),p(\cdot)) \in \overline{X}(b_i, b_{-i})$, each of which generates a payment and a value $(p_i(b_i, b_{-i}),v_i(x_i))$.  Take the collection of all such points, as plotted with payment on the $x$-axis and value on the $y$-axis, and define $V_i(p)$ to be the upper concave envelope of those points.  That is, $V_i(p)$ is the maximum possible value that buyer $i$ can obtain by paying $p$, possibly in expectation over some randomization.  (We will later show that each point on the concave envelope actually corresponds to a choice of bid and tie-breaking rule, so no randomization over bids is necessary.)  Note that $V_i$ is semi-differentiable, concave, weakly increasing, and $V_i(0) = 0$.  %

The maximum possible utility attainable by buyer $i$ when paying $p$ is $V_i(p) - C_i(p)$.  Since $V_i$ is concave and $C_i$ is convex, intuitively this quantity will be maximized at a value of $p$ for which the derivatives of $V_i$ and $C_i$ are equal, though this is not precise since $V_i$ and $C_i$ might be only semi-differentiable. 
More formally,
since $V_i$ is semi-differentiable its left and right derivatives exist; we let $Q_i^{-}(p)$ denote the left derivative of $V_i$ at $p$, and $Q_i^+(p)$ the right derivative.  For convenience, we set $Q_i^{-}(0) = \max\{1, Q_i^{+}(0)\}$.  We'll also write $R_i^-(p)$ and $R_i^+(p)$ for the left and right derivatives of $C_i$ at $p$.  For convenience we'll define $R_i^{-}(0) = 0$.  Since $V_i$ is concave and $C_i$ is convex, either the buyer's utility is maximized at a point $p$ such that $Q_i^-(p) \geq R_i^-(p)$ and $Q_i^+(p) \leq R_i^+(p)$ (which might be $p = 0$, recalling that $R_i^{-}(0) = 0$), or else it increases without bound as $p$ increases.  But the latter case cannot occur, since $V_i(p)$ is assumed to be bounded, whereas $C_i(p)$ increases without bound as $p \to \infty$ (since $R^+_i(p) \geq 1$ for all $p$).  So it must be that the utility is maximized at a finite value of $p$.  Take $r$ to be any value such that $Q_i^+(p) \leq r \leq Q_i^-(p)$ and $R_i^-(p) \leq r \leq R_i^+(p)$.  We can think of $r$ as a slope of curves $V_i$ and $C_i$ at $p$.  Also note that since $Q_i^{-}(0) \geq 1$, $R_i^{+}(0) \geq 1$, and $C_i$ is convex, we can assume without loss that $r \geq 1$.

In order to show that there is a utility-maximizing bid of the form $\alpha_i \cdot v_i$, it suffices to show that for any $r \geq 1$, each point on curve $V_i$ with a subderivative of $r$ corresponds to a bid of the form $\alpha_i \cdot v_i$ with $\alpha_i = 1/r$, paired with an appropriate tie-breaking rule.  
Choose some $r \geq 1$ and consider
an alternative outside option value curve $\widetilde{C}_i$ given by $\widetilde{C}_i(p) = r \cdot p$ for all $p$.  That is, $\widetilde{C}_i$ is a constant value of $r$ per unit of money spent.  Strategically, a buyer with valuation $v_i$ and outside option curve $\widetilde{C}_i$ behaves equivalently to a quasi-linear buyer with valuation $v_i / r$.  That is, we can re-scale the buyer's valuation to express it in units of money rather than units of utility.  But now, since the underlying mechanism is assumed to be truthful for quasi-linear buyers,
we know that such a buyer must maximize utility by declaring valuation $v_i / r$.  In other words, the uniform scaling bid $v_i / r$ must select the utility-maximizing point on curve $V_i(\cdot)$, which must be a point $p$ at which $V_i(p)$ has subderivative $r$ (as this is the derivative of $\widetilde{C}_i$ everywhere).

If there is a unique point on curve $V_i$ with subderivative $r$, say $(p, V_i(p))$, then we conclude that this must the point selected by the uniform scaling bid $v_i / r$.  That is, we must have $v_i(x_i(v_i / r, b_{-i})) = V_i(p)$ and $p_i(v_i/r, b_{-i}) = p$, as required.

Suppose instead that there are multiple points on curve $V_i(p)$ with the same subderivative $r \geq 1$.  By concavity, this means that they form a contiguous non-empty closed interval with prices in some range $[\beta,\gamma]$. 
We now argue that the points on $V_i$ in this interval correspond to different tie-breaking rules at input $(v_i / r, b_{-i})$.  For each sufficiently small $\epsilon > 0$, there is a point $(p^\epsilon, V_i(p^{\epsilon}))$ selected by uniform scaling bid $v_i / (r + \epsilon)$, and $V_i$ has subderivative $(r+\epsilon)$ at this point.  As $V_i$ is concave and continuous,  $p^\epsilon$ must be weakly increasing as $\epsilon \to 0$, and moreover $p^\epsilon \to \beta$ as $\epsilon \to 0$ (since $\lim_{\epsilon \to 0} Q_i^-(p^\epsilon) = \lim_{\epsilon \to 0} r+\epsilon = r$,  $Q_i$ is non-increasing, and therefore $p^\epsilon \to \rjc{\beta}$).
 The sequence of allocations and payments $(x_i( v_i / (r+\epsilon), b_{-i}), p_i( v_i / (r+\epsilon), b_{-i} ))$ as $\epsilon \to 0$ must have a convergent subsequence, say converging to some $(x^{(L)}, {\beta})$, and by definition this $(x^{(L)}, {\beta})$ lies in the convex closure $\overline{X}(v_i / r, b_{-i})$.  We conclude that the point $({\beta}, V_i({\beta}))$ is selected by input $v_i / r$ and the tie-breaking rule that returns outcome $x^{(L)}$ on input $(v_i/r, b_{-i})$.  A similar argument over the sequence of allocations $(x_i( v_i / (r-\epsilon)), b_{-i})$ as $\epsilon \to 0$ shows that point $({\gamma}, V_i({\gamma}))$ is selected by input $v_i / r$ and a different tie-breaking rule that selects some other outcome, say $x^{(H)}$.  Now choose any point in the interval {$[\beta,\gamma]$}, say $(z, V_i(z))$ for some $z = {\beta} + \lambda({\gamma}-{\beta})$ with $\lambda \in [0,1]$.  Then we can select this point with bid $v_i / r$ by using a tie-breaking rule that selects the convex combination $\lambda \cdot x^H + (1-\lambda) \cdot x^L$, which must lie in $\overline{X}(r \cdot v_i, b_{-i})$ due to convexity\footnote{{Recall that this might require randomization over outcomes, such as when goods are indivisible, but the bids can remain deterministic.}}.

We conclude that for any $p$ that maximizes $v_i(p) - C_i(p)$, and any $r$ such that $Q_i^+(p) \leq r \leq Q_i^-(p)$ and $R_i^-(p) \leq r \leq R_i^+(p)$, there exists a bid $b_i = \alpha_i \cdot v_i$ and tie-breaking rule $(x(\cdot), p(\cdot))$ for $M$ such that $r = 1/\alpha_i$ and $p = p_i(b)$.  Any such $b_i$ is ROI-optimal for agent $i$ by definition, and indeed this characterizes all ROI-optimal choices of $b_i$ and tie-breaking rule $x$ (by varying the choice of $r$), which are utility-maximizing from the choice of $p$.

Finally, since we assume that the (left and right) derivatives of $C_i$ are at least $1$ anywhere, we must have $r \geq 1$ at any utility-maximizing payment $p$, which implies that $\alpha_i = 1/r \leq 1$.
\end{proof}

A direct corollary of Theorem~\ref{thm:best.response.alpha.ties} is that if a bid profile $b$ is simultaneously ROI-optimal for all agents given some tie-breaking rule $x$, then it must be utility-maximizing for all agents and therefore forms a pure Nash equilibrium.

\begin{corollary}
\label{cor::ROI-implies-equil.ties}
Given a bid profile $b$ and tie-breaking rule $(x,p)$, if for every agent $i$ the bid $b_i$ is ROI optimal given $(x,p)$,
then $b$ is a pure Nash equilibrium for allocation and payment rules $(x,p)$.
\end{corollary}
}

\subsection{Equilibrium Existence}\label{sec:eq-exist.ties} 


We next wish to show that there is a tie-breaking rule such that the truthful mechanism has a pure Nash equilibrium 
 in which all agents use ROI-optimal uniform scaling strategies.
 As before, our definition of equilibrium is with respect to the broader space of bidding strategies in which each agent can declare any value for any item.  
Also note that the existence of a pure Nash equilibrium is not immediate, since the strategy space is infinite and payoffs may be discontinuous in bids.
An additional challenge in proving the existence of an equilibrium is that ties need to be broken by a single tie-breaking rule that is simultaneously optimal for all buyers, even though it only allocates the available items (it never over-allocates).

%
%
Another potential challenge is that there might exist pure Nash equilibria in uniform scaling strategies that are not ROI-optimal.  For example, if there are two quasi-linear buyers and a second-price auction for a single good, with $v_1 = 1$ and $v_2 = 2$, then $(\alpha_1, \alpha_2) = (0,1)$, with its corresponding allocation $(x_1, x_2) = (0,1)$, is an equilibrium even though ${1/R^+_1(0)} = 1 > 0 = \alpha_1$.  This equilibrium is driven by indifference: buyer $1$ obtains the same utility at any declared $\alpha_1\leq 1$, and she shades her bid more than necessary in the face of this indifference.  
More generally, the only way that buyer $i$ could have $\alpha_i < 1 / R^+_i(p_i(\alpha))$ at an equilibrium is if she is indifferent between bidding $\alpha_i$ and some strictly larger bid, due to the allocation not changing, but \emph{would have wanted} to spend a positive amount of money, to obtain value $1 / \alpha_i$ per unit, if doing so were possible.  This cannot occur in an ROI-optimal equilibrium, since each buyer's bid is at least the marginal price at which they would be willing to spend additional money. Note that we side-stepped this issue in Section~\ref{sec:equil.exact} by explicitly assuming that each agent's payment is strictly increasing in $\alpha_i$, which implied that any utility-maximizing choice of $\alpha_i$ must be ROI-optimal.

We can now prove that an ROI-optimal equilibrium always exists under appropriate tie-breaking.

\begin{theorem}
\label{thm:equil.proper.ties2}
Fix any mechanism $M$ that is truthful for quasi-linear buyers.
Then for any given market instance there exists a tie-breaking rule consistent with $M$ for which there exists a ROI-optimal pure Nash equilibrium in uniform scaling strategies.
\end{theorem}

Our strategy for proving Theorem~\ref{thm:equil.proper.ties2} 
is to adjust the auction game with carefully chosen small perturbations.  The adjusted allocation and payment rules will be continuous and the mechanism will remain truthful, so we can use the argument from Section~\ref{sec:eq-exist} to show that pure Nash equilibria exist for the adjusted game.  Moreover, in the limit as these perturbations vanish, the equilibria converge to the desired ROI-optimal equilibrium in the original game.

Recall that in the game implied by uniform scalings for truthful mechanism $M$, agents declare multipliers $\alpha = (\alpha_i) \in [0,1]^n$, which are used to generate scaled valuations $b_{i} = \alpha_i \cdot v_{i}$ which are then passed to mechanism $M$.  We will modify this game as follows, parameterized by some $\delta > 0$ (which we think of as being sufficiently small):
\begin{enumerate}
\item We adjust each agent's cost of money curve to relax hard budgets.  Agent $i$'s disutility for total payment $p$ will be $C_i^{(\delta)}(p)$, where
\[ C^{(\delta)}_i(p) = \begin{cases}C_i(p) & \text{ if } p \leq B_i,\\ 
C_i(B_i) + (p - B_i) \cdot \tfrac{1}{\delta} & 
\text{ if } p > B_i.\end{cases} \] 
\item For each agent $i$ we introduce an additional good.  Agent $i$ has value $\delta$ for this good, and all other agents have value $0$ for it.  Any value obtained by winning this good combines additively with values obtained from the original outcome space.  This good is subject to a random reserve price, chosen uniformly from $[0,2\delta]$, and agent $i$ does not know the value of this reserve price when it bids.  The agent wins the good if they declare a value for it that is strictly greater than the realized reserve, in which case they pay the reserve price.
\item For each agent $i$, a value $\gamma_i$ is drawn uniformly at random from $[1-\delta, 1]$.\footnote{In fact, we can choose any atomless distribution supported on $[1-\delta,1]$, but our proofs will assume uniformity for notational convenience.}
Agent $i$'s valuation is then perturbed to be $\gamma_i \cdot v_i$.  That is, the utility that agent $i$ obtains from an outcome $x_i$ is taken to be $\gamma_i \cdot v_i(x_i)$.  Moreover, agent $i$'s bid is taken to be $\gamma_i \cdot \alpha_i \cdot v_{i}$.  This perturbation applies also to the additional good added for each agent.
\end{enumerate}

\textbf{Comment} When $B_i$ is finite, to ensure that $ C^{(\delta)}_i(p)$ is convex for $p>B_i$, we need $\delta$ to be small enough to satisfy $1/\delta \geq R^-(B_i)$.
We call this the $\delta$-perturbed game.  In general, we use $u_i^{(\delta)}$, $x_i^{(\delta)}$, etc., to denote outcomes of this perturbed game.  
The motivation for these perturbations is the following lemma, which establishes
the existence of a pure Nash equilibrium in the perturbed auction game.

\begin{lemma}
\label{lem:pure.equil.ties}
For all sufficiently small positive $\delta$,
there exists a pure Nash equilibrium $\alpha^{(\delta)}$ of the $\delta$-perturbed game satisfying $\alpha_i\leq 1$ for every $i$.  This pure Nash equilibrium is ROI-optimal with respect to the perturbed curves $C_i^{(\delta)}$.
\end{lemma}
\begin{proof}
We will argue that the perturbed game satisfies the technical assumptions imposed in Section~\ref{sec:eq-exist}, and thus the simplified proof of Theorem~\ref{thm:equil.proper.general} in that section implies the desired result.

We first note that the perturbed game is truthful for quasi-linear agents (that is, agents for which $C_i(p) = p$ for all $p$).  This follows because the original mechanism $M$ is truthful, and hence it is a weakly dominant strategy to truthfully report one's valuation over outcomes of $M$ if that valuation is $\gamma_i \cdot v_i$ (since the space of valuations is assumed to be closed under scalar multiples).  It is likewise a dominant strategy to report one's value truthfully for the additional good with random reserve price.  Thus, for each realization of $\gamma$, it would be a weakly dominant strategy to report $\gamma_i \cdot v_i$ if one were reporting valuations directly.  It is thus weakly dominant to report $v_i$ to the perturbed game, which implements the scaling by $\gamma_i$.

Second, for sufficiently small $\delta$, $C^{(\delta)}_i$ is convex, increasing, and is not subject to a hard budget constraint.

Third, the outcomes and payments in the perturbed game are continuous functions of $\alpha_i \in [0,1]$.  This follows from the fact that $\gamma_i$ is drawn from an atomless distribution: for any sufficiently small $\mu > 0$, as we let $\mu \to 0$, the distribution over outcomes and payments observed by declaring {$\alpha_i+\mu_i$} converges in total variation distance to the distribution over outcomes and payments observed by declaring {$\alpha_i$}.  In more detail: an agent who reports $\alpha_i$ will observe a uniform distribution over outcomes corresponding to uniform scaling bids lying in $[(1-\delta)\alpha_iv_i, \alpha_i v_i]$ and an agent who reports $(\alpha_i + \mu)$ will observe a uniform distribution over outcomes corresponding to uniform scaling bids lying in $[(1-\delta)(\alpha_i+\mu)v_i, (\alpha_i +\mu)v_i]$.  As $\mu \to 0$ these distributions over bids converge in the sense that they can be coupled with equality occurring with probability approaching $1$.  The observed outcomes can therefore also be coupled such that equality occurs with probability approaching $1$, and hence the resulting value and payments likewise converge as $\mu \to 0$.

Finally, each agent's expected payment, excluding the payment for the additional good, is weakly increasing in their declared multiplier $\alpha_i$ due to truthfulness.  The expected payment for the additional good is strictly increasing, due to the random reserve.  We conclude that the total payment of agent $i$ is strictly increasing in the declared multiplier $\alpha_i$.

Our game therefore satisfies the technical assumptions imposed in Section~\ref{sec:eq-exist}.  Thus, by the proof of Theorem~\ref{thm:equil.proper.general} presented in Section~\ref{sec:eq-exist}, an ROI-optimal pure Nash Equilibrium in uniform scaling strategies exists in the perturbed game, for all sufficiently small $\delta$.
\end{proof}

For each sufficiently small $\delta$, let $\alpha^{(\delta)}$ be a pure Nash equilibrium from the statement of Lemma~\ref{lem:pure.equil.ties}.
Since each $(\alpha^{(\delta)}, x^{(\delta)}(\alpha_i^{(\delta)}), p^{(\delta)}(\alpha_i^{(\delta)}))$ lies in a compact space (as $0\leq \alpha_i^{(\delta)} \leq 1$ for all $i$), there is a subsequence 
that converges to some $(\alpha^*, x^*, p^*)$ as $\delta \to 0$.
We wish to show that $(\alpha^*, x^*)$ is an ROI-optimal competitive auction equilibrium for the (unperturbed) auction, completing the proof of Theorem~\ref{thm:equil.proper.ties2}.


\begin{lemma}
There is a tie-breaking rule for mechanism $M$ that returns outcome $x^*$ and payments $p^*$ on input $\alpha^* \cdot v$, such that 
$\alpha^*$ is an ROI-optimal
pure Nash equilibrium under this tie-breaking rule.
\end{lemma}
\begin{proof}
We first show that $(\alpha^*, x^*, p^*)$ is a valid auction outcome for mechanism $M$ under some tie-breaking rule.  First, since $\alpha^{(\delta)}_i \leq 1$
for each $\delta$, taking the limit as $\delta$ goes to zero yields $\alpha^*_i \leq 1$ for each $i$.  Moreover, for any fixed $\delta$, the allocation and payments $(x^{(\delta)}, p^{(\delta)})$ is a convex combination of outcomes of $M$ under uniform scaling bids with scalars lying in $[(1-\delta)\alpha_i^{(\delta)}, \alpha_i^{(\delta)}]$ for each $i$.  As $\delta \to 0$, this range converges to the singleton $\alpha^*$, and the resulting outcomes and payments converge to $(x^*, p^*)$.   Therefore $(x^*, p^*)$ is a convex combination of allocations obtained by sequences of {bids} 
converging to $\alpha^* \cdot v$, and therefore by definition $(x^*, p^*)$ lies in the convex closure of the allocation rule for $M$ on input $\alpha^* \cdot v$.  In particular, $(x^*, p^*)$ is a valid tie-breaking outcome for $M$ on input $\alpha^* \cdot v$.

Finally, we note that 
\[ \lim\sup_{\delta \to 0} R^{(\delta)+}_i(p^{(\delta)}_i(\alpha^{(\delta)})) 
\leq \lim\sup_{\delta \to 0} R^{(\delta)+}_i(p_i(\alpha^*,x^*))
= R^+_i(p_i(\alpha^*, x^*))\]
where the inequality follows because $\lim_{\delta \to 0}p^{(\delta)}_i(\alpha^{(\delta)}) = p_i(\alpha^*, x^*)$ and $R^{(\delta) +}_i$ is right-continuous (recalling that right
continuity implies the limit value is upper bounded by the value at the limit),
and the equality uses the fact that $\lim_{\delta \to 0}R^{(\delta)+}_i(p) = R_i^{+}(p)$ for all $p$.
On the other hand, Lemma~\ref{lem:pure.equil.ties} implies that equilibrium $\alpha_i^{(\delta)}$ is ROI-optimal (with respect to $C_i^{(\delta)}$) for all $\delta$; {therefore} $\alpha_i^{(\delta)} \geq 1 / R^{(\delta) +}_i(p^{(\delta)}_i(\alpha^{(\delta)}))$ for all $i$.  As $\alpha_i^{(\delta)} \to \alpha_i^*$ as $\delta \to 0$, we conclude that $\alpha_i^* \geq 1 / R^+_i(p_i(\alpha^*, x^*))$.
Similarly (by considering $\lim\inf_{\delta \to 0} R^{(\delta)-}_i(p^{(\delta)}_i(\alpha^{(\delta)}))$), we obtain $\alpha_i^* \leq 1 / R^-_i(p_i(\alpha^*, x^*))$).  We conclude that $\alpha^*$ is ROI-optimal for all $i$ under the tie-breaking rule that selects outcome $x^*$ and payments $p^*$, and thus by Corollary~\ref{cor::ROI-implies-equil.ties} it is an ROI-optimal pure Nash equilibrium, as required.

\end{proof}

%
%


\end{document}